\documentclass[a4paper,11pt,reqno]{amsart}

\usepackage{hyperref}       
\usepackage{url}            
\usepackage{amsthm}
\usepackage{amsfonts}
\usepackage{amsmath}
\usepackage{amssymb}
\usepackage{multirow}
\usepackage{float}
\usepackage[noadjust]{cite}

\usepackage[margin=2.7cm]{geometry}

\usepackage{tikz}
\usepackage{caption}  
\usepackage{ltablex}\keepXColumns
\newcolumntype{Y}{>{\centering\arraybackslash}X}

\allowdisplaybreaks

\theoremstyle{definition}

\newtheorem{theorem}{Theorem}[section]
\newtheorem{corollary}[theorem]{Corollary}
\newtheorem{lemma}[theorem]{Lemma}
\newtheorem{definition}[theorem]{Definition}

\newtheorem{remark}[theorem]{Remark}
\newtheorem{notation}[theorem]{Notation}

\newtheorem{example}[theorem]{Example}

\DeclareMathOperator\rowsp{rowsp}
\DeclareMathOperator\colsp{colsp}
\DeclareMathOperator\tr{Tr}
\def\Fqnm{\mathbb{F}_q^{n\times m}}
\def\Fqnn{\mathbb{F}_q^{n\times n}}

\DeclareMathOperator\maxrank{maxrk}
\DeclareMathOperator\rank{rk}
\def\Fqm{\mathbb{F}_{q^m}}
\def\Fq{\mathbb{F}_q}
\def\Fqn{\mathbb{F}_q^n}
\def\dimq#1{\dim_{\mathbb{F}_q}(#1)}
\DeclareMathOperator\supp{supp}
\DeclareMathOperator\colsupp{colsupp}
\DeclareMathOperator\rowsupp{rowsupp}
\newcommand{\qbin}[2]{\begin{bmatrix}{#1}\\ {#2}\end{bmatrix}_q}
\newcommand{\qbintwo}[2]{\begin{bmatrix}{#1}\\ {#2}\end{bmatrix}_2}
\newcommand{\bin}[2]{\begin{pmatrix}{#1}\\ {#2}\end{pmatrix}}
\newcommand{\qmbin}[2]{\begin{bmatrix}{#1}\\ {#2}\end{bmatrix}_{q^m}}

\DeclareMathOperator\gl{GL}

\DeclareMathOperator\wt{wt}
\def\mC{\mathcal{C}}

\title{Rank-Metric Codes, Generalized Binomial Moments and their Zeta Functions}
\author{Eimear Byrne}
\author{Giuseppe Cotardo$^*$}
\thanks{$^*$The author was supported by the Irish Research Council, grant n. GOIPG/2018/2534.}
\author{Alberto Ravagnani$^{**}$}
\thanks{$^{**}$The author was supported by the Marie Curie Research Grants
Scheme, grant n. 740880.}

\begin{document}

	\maketitle
\begin{abstract}
	In this paper we introduce a new class of extremal codes, namely the $i$-BMD codes. We show that for this family 
several of the invariants are determined by the parameters of the underlying code. We refine and extend the notion of an $i$-MRD code and 
show that the $i$-BMD codes form a proper subclass of the $i$-MRD codes.
Using the class of $i$-BMD codes we then obtain a relation between the generalized rank weight enumerator and its corresponding generalized zeta function. We also establish a MacWilliams identity for generalized rank weight distributions. 

\end{abstract}

\section{Introduction}

A well-studied problem in coding theory is the determination of the invariants of a code, such as its rank distribution, its binomial moments and its generalized weights. 
Computation of such invariants for an arbitrary code is a non-trivial problem. On the other hand, for codes in certain families, all or some of these invariants are determined by the standard coding theoretic parameters of length, dimension and minimum distance.

Optimal codes are of great interest in coding theory and often have rigidity properties, which make them interesting as combinatorial objects. 
Perhaps the best known class of optimal codes in the {\em rank metric} are the \textit{maximum rank distance} (MRD) codes. This family were first introduced in the coding theory literature by Delsarte~\cite{delsarte1978bilinear} and are characterised as being the $k$-dimensional subspaces of $\Fqnm$ that attain the rank-metric analogue of the Singleton bound. In this sense, the class of MRD codes may be regarded as a $q$-analogue of the {\em maximum distance separable} (MDS) linear block codes.
Delsarte's construction (see also~\cite{gabidulin1985theory, roth1991maximum}) immediately yields the existence of MRD codes for any choice of the parameters $q,m,n$ and minimum rank distance. Moreover, these parameters fully determine all the invariants of the codes in this class. 
Similarly, the weight enumerator of an MDS code is determined by its parameters $q,n$ and minimum Hamming distance.

An interesting fact about the MRD codes is that for a fixed ambient space $\Fqnm$ (where we write $n\leq m$, without loss of generality),  
the set of $n+1$ distinct MRD rank weight enumerators forms a ${\mathbb Q}$-basis of the space of homogeneous polynomials in ${\mathbb Q}[X,Y]$ of degree $n$. 
The analogous statement holds also for the MDS weight enumerators. 
These observations have been exploited in studying the {\em zeta function} of a linear code.
This object was introduced in~\cite{duursma1999weight,duursma2001weight} for linear block (Hamming metric) codes and in~\cite{blanco2018rank} for rank-metric codes. 
The zeta function is the generating function of the \textit{normalized binomial moments} of a code and can be related to the code's weight enumerator.
The corresponding recurrence formula of the zeta function, the {\em zeta polynomial}, turns out to have coefficients that are exactly the coefficients that arise in the expression of the weight enumerator
as a ${\mathbb Q}$-linear combination of Singleton-optimal weight enumerators (MDS for the Hamming metric and MRD for the rank metric). 

A central problem studied in this paper is on the behaviour of the generalized rank weight distribution of a code with respect to its generalized binomial moments and on the connections between these objects via zeta functions. In order to develop such a theory, what is first required is the correct notion of code optimality, as optimal codes provide the fundamental building blocks for this theory. 

We mention some classes of optimal codes, for example, the $\Fqm$-linear $i$-MRD codes~\cite{ducoat2014rank} and the \textit{dually quasi-MRD} (DQMRD) codes~\cite{de2018weight}. 
In~\cite{de2018dually}, de la Cruz introduced another class of $\Fq$-linear rank-metric codes that are $i$-MRD for $i\in\{1,1+m,\ldots,1+(\left\lceil\frac{k}{m}\right\rceil-1)m\}$.
It turns out that the $i$-MRD property does not provide a class of codes equipped to describe generalized zeta functions. For this reason, we introduce a new class of extremal codes in the rank metric, namely the family of $i$-BMD (\textit{binomial moment determined}) codes, which we will see are a subclass of the $i$-MRD codes. This is another class of codes whose invariants are determined. 
In this paper we introduce the zeta function for generalized rank weights and we show that using the $i$-BMD property, we can extend the connection between the \textit{generalized zeta functions} and the \textit{generalized rank weight enumerators} of a code. 
We also generalize the notion of $i$-MRD for the remaining cases. We study these new objects via an anticode approach. 
Introduced in~\cite{ravagnani2016generalized}, this technique gives a more general analysis of the theory and allows us to easily generalize the results in this paper for codes in other metrics, such as the Hamming metric. The MacWilliams identities~\cite{delsarte1978bilinear,ravagnani2016rank} give an explicit way to compute the binomial moments and the rank distribution of a code from those of its dual. 
We describe these identities for generalized rank weights.

\subsection*{Outline.} \ 
In Section \ref{SectionPrelim} we recall some well-known definitions and results. In Section \ref{SectionBin} we refine and extend, via an anticode approach, the definition of binomial moments introduced in~\cite{blanco2018rank} and we link them with the generalized rank weight distribution. 

In Section \ref{SectioniMRD} we introduce the notion of an $i$-BMD code. We show that, due to their strong rigidity properties, $i$-BMD codes allow us to determine a priori, for all $i\leq j\leq k$, their $j$-th generalized rank weights, their $j$-th generalized binomial moments and their $j$-th generalized rank distributions. We also extend the definition of the class of $i$-MRD codes. 

In Section \ref{SectionZeta} we describe the $i$-th generalized zeta function of a code and 
relate this to the $i$-th generalized rank weight enumerator. We show that the $j$-th generalized rank weight enumerators of the $i$-BMD codes, $i\leq j\leq k$, form a $\mathbb{Q}$-basis for the space of the $j$-th generalized rank weight enumerators. We give an explicit formula to compute the coefficient of a $j$-th generalized rank weight enumerator with respect to this basis in Section \ref{SectionBell}, using the well-known Bell polynomials.

In Section \ref{SectionDuality} we derive the MacWilliams identities for generalized rank weight distributions and we show how to explicitly compute the $i$-th generalized binomial moments ($i$-th generalized normalized binomial moments resp.) of a code, knowing all the $j$-th generalized binomial moment ($j$-th generalized normalized binomial moments resp.), $0\leq j\leq i$, of its dual. We then use these results to compute the $i$-th generalized rank weight distribution and the $i$-th generalized zeta function. Finally, in Section \ref{SectionHamming} we describe $i$-BMD codes for the Hamming codes and we prove that they indeed coincide with the class of $i$-MDS codes.

\section{Preliminaries}
\label{SectionPrelim}
Throughout the paper, $q$ is a prime power and $\mathbb{F}_q$ is the finite field with $q$ elements. We let $n,m$ be positive integers and assume $2\leq n \leq m$ without loss of generality. We denote the row-space and column-space of a matrix $M\in\Fqnm$  by $\rowsp(M)$ and $\colsp(M)$ respectively. 

\begin{definition}
	A (\textbf{matrix} \textbf{rank-metric}) \textbf{code} is a subspace $\mC \le \Fqnm$. The \textbf{maximum rank} of $\mathcal{C}$ is $\maxrank(\mathcal{C}):=\max\{\rank(M):M\in \mathcal{C}\}$. The \textbf{minimum} (\textbf{rank}) \textbf{distance} of a non-zero code $\mathcal{C}$ is $d(\mathcal{C}):=\min\{\rank(M):M\in \mathcal{C},\; M \neq 0\}$.
We also define the minimum distance of $\mathcal{C}=\{0\}$ to be $n+1$ following \cite[Definition~3.1]{gorla2019rank}. 
\end{definition}

From now on, unless otherwise stated, $\mathcal{C}\leq \Fqnm$ is a rank-metric code whose dimension is denoted by $k$
and its minimum distance by $d$.

\begin{definition}
	The \textbf{dual} of  $\mathcal{C}$ is the code
	\begin{equation*}
		\mathcal{C}^\perp:=\{N\in\Fqnm: \tr(MN^t)=0\; \textup{ for all } M\in\mathcal{C}\} \le \Fqnm.
	\end{equation*}
	where $\tr(MN^t)$ is the trace of the square matrix $MN^t$.
\end{definition}
	
We denote by $k^\perp$ and $d^\perp$ the dimension and the minimum distance of $\mathcal{C}^\perp$ respectively.

\begin{definition}
	The \textbf{row-support} and the \textbf{column-support} of $\mathcal{C}$ are 
	\begin{equation*}
		\colsupp(\mathcal{C}):=\sum_{M\in\mathcal{C}}\colsp(M) \quad \mbox{and} \quad  \rowsupp(\mathcal{C}):=\sum_{M\in\mathcal{C}}\rowsp(M).
	\end{equation*} 
Note that $\colsupp(\mathcal{C})$ and $\rowsupp(\mathcal{C})$ are subspaces of $\Fqn$ and $\mathbb{F}_{q}^m$ respectively.
\end{definition}

\begin{definition}
	Let $U$ be a subspace of $\Fqn$. The subcodes of $\mathcal{C}$ \textbf{column-supported} and \textbf{row-supported} 
on $U$ are 
	\begin{equation*}
		 \mathcal{C}(U):=\{M\in\mathcal{C}:\colsp(M)\leq U\} \quad \mbox{and} \quad \mathcal{C}[U]:=\{M\in\mathcal{C}:\rowsp(M)\leq U\},
	\end{equation*}
respectively.
\end{definition}

A rank-metric analogue of the Singleton bound for a rank-metric code $\mathcal{C}$  was proved by Delsarte in \cite[Theorem~5.4]{delsarte1978bilinear} and it can be stated as 
\begin{equation}
\label{RankSingletonBound}
	k\leq m(n-d+1).
\end{equation}

We say that $\mathcal{C}$ is an \textbf{MRD} (\textbf{Maximum Rank Distance}) code if it meets the bound in~\eqref{RankSingletonBound}. One can easily check that the code $\{0\}$ and its dual $\Fqnm$ are MRD. In~\cite[Theorem~5.5]{delsarte1978bilinear}, Delsarte proved that $\mathcal{C}$ is MRD if and only if its dual code $\mathcal{C}^\perp$ is MRD and that such codes exist for every choice of the parameter $n,m$ and $d$. In \cite[Proposition~47]{ravagnani2016rank} another upper bound on the dimension of $\mathcal{C}$ was given as:
\begin{equation}
	\label{RankAnticodeBound}
	k\leq m\cdot\maxrank(\mathcal{C}).
\end{equation}

\begin{definition}
[{\cite[Definition~22]{ravagnani2016generalized}}] We say that
	$\mathcal{C}$ is an \textbf{optimal anticode} if it attains the bound in~\eqref{RankAnticodeBound}.
\end{definition}
	
We denote by $\mathcal{A}$ the set of optimal anticodes in $\Fqnm$ and by $\mathcal{A}_u$ the set of $mu$-dimensional optimal anticodes, for any $0\leq u\leq n$. In particular,
	\begin{equation*}
		\mathcal{A}=\bigsqcup_{u=0}^n\mathcal{A}_u.
	\end{equation*}

It is known that $\mathcal{C}$ is an optimal anticode if and only if $\mathcal{C}^\perp$ is an optimal anticode~\cite{ravagnani2016rank}. Optimal anticodes were characterized by Meshulam~\cite[Theorem~3]{meshulam1985maximal}, who gave a proof for the square case~$n=m$, but from which the case $n<m$ easily follows.

\begin{theorem}
\label{TheoremAnticodeEquiv}
The following hold.
	\begin{itemize}
		\item $\mathcal{A}=\{\Fqnm(U): U\leq\Fqn \}$, if $n<m$;
		%\item $\mathcal{A}=\{\Fqnm[U]: U\leq\Fqn\}$, if $n>m$;
		\item $\mathcal{A}=\{\Fqnm(U): U\leq\Fqn\}\cup\{\Fqnm[U]: U\leq\Fqn\}$, if $n=m$.
	\end{itemize}
\end{theorem}

Note moreover that for all $n$ and $m$ and all $0 \le u \le n$ we have $\mathcal{A}_u \neq \emptyset$. 

\begin{lemma}
\label{LemmaAntIntersection}
	Suppose $n=m$ and let $1\leq u\leq n-1$. We have
	\begin{equation*}
		\{\Fqnn(U): U\leq\Fqn, \; \dimq{U}=u\}\cap\{\Fqnn[U]: U\leq\Fqn, \; \dimq{U}=u\}=\emptyset.
	\end{equation*}
\end{lemma}
\begin{proof}
	Suppose, toward a contradiction, that $U,V\leq\Fqnn$ satisfy $\Fqnn(U)=\Fqnn[V]$ and $\dimq{U}=\dimq{V}=u$. It is not hard to see that there exist matrices $A,B\in\gl_n(\Fq)$ such that
		\begin{equation*}
			A\cdot\Fqnn(U)=\Fqnn\left(\left<e_1,\ldots,e_u\right>_q\right) \quad \mbox{and} \quad \Fqnn[V]\cdot B=\Fqnn\left[\left<e_1,\ldots,e_u\right>_q\right],
		\end{equation*}
where $e_i$ denotes the $i$-th element of the standard basis of $\mathbb{F}_q^n$.
		Indeed, we can take as $A$ and $B$ the matrix representations of any $\Fq$-isomorphisms $f$, $g$ of $\Fq^n$
such that $f(U)=\left<e_1,\ldots,e_u\right>_q$ and $g(V)=\left<e_1,\ldots,e_u\right>_q$. We then have
		\begin{equation*}
			\Fqnn(U)\cdot B=\Fqnn(U) \quad \mbox{and} \quad A\cdot \Fqnn[V] =\Fqnn[V].
		\end{equation*}
		Therefore,
		\begin{equation*}
			\Fqnm\left(\left<e_1,\ldots,e_u\right>_q\right)=A\cdot\Fqnm(U)\cdot B=A\cdot\Fqnm[V]\cdot B=\Fqnm\left[\left<e_1,\ldots,e_u\right>_q\right],
		\end{equation*}
		which is impossible as $1 \le u \le n-1$.
\end{proof}

We recall the definition of generalized rank weights introduced in \cite{ravagnani2016generalized}. See~\cite[Section~5]{gorla2019rank}
for an overview of the alternative definitions and characterizations that have been proposed.

\begin{definition}
[{\cite[Definition~23]{ravagnani2016generalized}}]
	For $0\leq i\leq k$, the \textbf{$i$-th generalized rank weight} of~$\mathcal{C}$ is 
	\begin{equation*}
		d_i(\mathcal{C}):=\frac{1}{m}\min\{\dimq{A}:A\in\mathcal{A},\; \dimq{A\cap\mathcal{C}}\geq i\}.
	\end{equation*}
\end{definition}

In~\cite{ravagnani2016generalized}, the $d_i(\mathcal{C})$s were called generalized Delsarte weights. When the code $\mathcal{C}$ is clear from context, we simply write $d_i$ for $d_i(\mathcal{C})$ and $d_i^\perp$ for $d_i(\mathcal{C}^\perp)$.

\begin{lemma}[{\cite[Theorem~30]{ravagnani2016generalized}}]
\label{LemmaGenRankWeights}
The following hold:
	\begin{itemize}
		\item $d_1=d$ if $\mathcal{C} \neq \{0\}$;
		\item $d_k\leq n$;
		\item $d_i\leq d_{i+1}$ for any $1\leq i\leq k-1$;
		\item $d_i< d_{i+m}$ for any $1\leq i\leq k-m$;
		\item $d_i\leq n-\left\lfloor\frac{k-i}{m}\right\rfloor$ for any $1\leq i\leq k$;
		\item $d_i\geq \left\lceil\frac{i}{m}\right\rceil$ for any $1\leq i\leq k$. 
	\end{itemize}
\end{lemma}

Another family of cardinality-extremal codes was introduced in~\cite{de2018weight}. The authors show that such codes exist for every choice of $n,m$ and $d$.

\begin{definition}
[{\cite[Definition~10]{de2018weight}}]
We say that	$\mathcal{C}$ is \textbf{QMRD} (\textbf{quasi-MRD}) if $m\nmid k$ and 
	\begin{equation*}
		d=n-\left\lfloor\frac{k-1}{m}\right\rfloor=n-\left\lceil\frac{k}{m}\right\rceil+1.
	\end{equation*}
	We say that $\mathcal{C}$ is \textbf{DQMRD} (\textbf{dually QMRD}) if both $\mathcal{C}$ and $\mathcal{C}^\perp$ are QMRD.
\end{definition}

Finally, we recall the definition and some well-know properties of the $q$-binomial coefficient; a standard reference is~\cite{andrews1998theory}. 

\begin{definition}
	Let $a,b$ be integers. The $q$-binomial coefficient of $a$ and $b$ is
	\begin{equation*}
		\qbin{a}{b}=
		\begin{cases}
			\displaystyle0 & \textup{ if }\; b<0 \;\textup{ or }\; 0\leq a\leq b,\\
			\displaystyle1 & \textup{ if }\; b=0 \;\textup{ and }\; a\geq 0,\\
			\displaystyle\prod_{i=1}^b\frac{q^{a-i+1}-1}{q^i-1} & \textup{ if }\; b> 0 \;\textup{ and }\; a\geq b,\\
			\displaystyle(-1)^bq^{ab-\binom{b}{2}}\qbin{-a+b-1}{b} & \textup{ if }\; b> 0 \;\textup{ and }\; a< 0.\\
		\end{cases}
	\end{equation*}
\end{definition}

\begin{lemma}
\label{LemmaTools}
	Let $a,b,c$ be integers. The following hold. 
	\begin{enumerate}
		\item $\qbin{a}{b}\qbin{b}{c}=\qbin{a}{c}\qbin{a-c}{a-b}$,\quad for $a,b,c\geq 0$;
		\item $\displaystyle\sum_{j=0}^c\qbin{c}{j}(-1)^iq^{\binom{j}{2}}a^{c-j}b^j=
		\begin{cases}\displaystyle
			1 & c=0,\\
			\displaystyle\prod_{j=0}^{c-1}(a-q^jb) & c\geq 1;
		\end{cases}$
		\item $\displaystyle\qbin{a+b}{c}=\sum_{j=0}^cq^{j(b-c+j)}\qbin{a}{j}\qbin{b}{c-j}=\sum_{j=0}^cq^{(c-j)(a-j)}\qbin{a}{j}\qbin{b}{c-j}$.
	\end{enumerate}
\end{lemma}

\section{Generalized Binomial Moments and Rank Distributions}
\label{SectionBin}
In this section, we define the generalized rank weight distributions and generalized binomial moments of a rank-metric code via the anticode approach. Our results extend those of~\cite{blanco2018rank}. Throughout the paper, $i$ is an integer in $\{0,1,\ldots,k\}$, unless otherwise stated.

\begin{lemma}[{\cite[Lemma~28]{ravagnani2016rank}}]
\label{LemmaCcapA}
	Let $A\in\mathcal{A}_u$. We have
	\begin{equation*}
		|\mathcal{C}\cap A|= |\mathcal{C}^\perp\cap A^\perp| \, q^{k-m(n-u)}.
	\end{equation*}
\end{lemma}

\begin{lemma}
\label{LemmaDimCcapA}
Let $A\in \mathcal{A}_u$. We have
	\begin{equation*}
		\dimq{\mathcal{C}\cap A}=
		\begin{cases}
			0 & \textup{ if }\; u<d,\\
			k-m(n-u) & \textup{ if }\; u>n-d^\perp.
		\end{cases}
	\end{equation*}
\end{lemma}
\begin{proof}
The result is straightforward if $\mathcal{C}=\{0\}$. We henceforth assume $k \ge 1$, which implies $1 \le d \le n$.
If $u<d$, then clearly $\mathcal{C} \cap A =\{0\}$. Similarly, if $n-u < d^\perp$, then $\mathcal{C}^\perp \cap A^\perp=\{0\}$, since
$A^\perp \in \mathcal{A}_{n-u}$ by~\cite[Theorem~54]{ravagnani2016rank}. Therefore by Lemma \ref{LemmaCcapA} we conclude
$|\mathcal{C}\cap A|= q^{k-m(n-u)}$, from which the result follows.
\end{proof}

\begin{definition}
\label{DefinitionBu}
	For $0\leq u\leq n$, the \textbf{$(u,i)$-th generalized binomial moment} $B_u^{(i)}(\mathcal{C})$ is 	
	\begin{equation*}
		 B_u^{(i)}(\mathcal{C}):=\begin{cases}
			\displaystyle\sum_{\tiny\begin{matrix}U\leq\Fqn\\\dimq{U}=u\end{matrix}} B_{U}^{(i)}(\mathcal{C},c) &\; \textup{ if }\; n<m,\\\\
			\displaystyle\frac{1}{2}\sum_{\tiny\begin{matrix}U\leq\Fqn\\\dimq{U}=u\end{matrix}} \left(B_{U}^{(i)}(\mathcal{C},c)+B_{U}^{(i)}(\mathcal{C},r)\right) &\; \textup{ if }\; n=m,
		\end{cases}
	\end{equation*}
	where, for every $U\leq\Fqn$, 
	\begin{equation*}
		 B_U^{(i)}(\mathcal{C},c):=\qbin{\dimq{\mathcal{C}(U)}}{i} \qquad\textup{ and }\qquad  B_U^{(i)}(\mathcal{C},r):=\qbin{\dimq{\mathcal{C}[U]}}{i}.	
	\end{equation*}
\end{definition}

Notice that this definition is required in order for the generalized binomial moment to be an invariant under duality.

\begin{lemma}
\label{LemmaBu}
	The following holds for any $0\leq u\leq n$.
	\begin{equation*}
		B_u^{(i)}(\mathcal{C})=
		\begin{cases}
			0 & \textup{ if }\; u<d_i,\\
			\qbin{n}{u}
		\qbin{k-m(n-u)}{i} & \textup{ if }\; u>n-d^\perp.
		\end{cases}
	\end{equation*}
\end{lemma}
\begin{proof}
Suppose $i=0$. Then $d_0=0$ and we only need to treat the case $u>n-d^\perp$. In such a case we have
\begin{equation*}
	B_u^{(0)}(\mathcal{C},c)=\qbin{\dimq{\mathcal{C}(U)}}{0}=1 \quad \mbox{and} \quad B_u^{(0)}(\mathcal{C},r)=\qbin{\dimq{\mathcal{C}[U]}}{0}=1,
\end{equation*}
from which
\begin{equation*}
	B_u^{(0)}=\sum_{\tiny\begin{matrix}U\leq\Fqn\\\dimq{U}=u\end{matrix}}1=\qbin{n}{u}=\qbin{n}{u}\qbin{k-m(n-u)}{0}.
\end{equation*}

Now suppose $i\ne 0$. We continue the proof assuming $n=m$. The case $n<m$ is analogous and in fact simpler. 

Suppose~$u>n-d^\perp$. Applying Lemma~\ref{LemmaDimCcapA} to Definition~\ref{DefinitionBu} we get
\begin{equation*}
\begin{aligned}
	B_u^{(i)}(\mathcal{C})&=\frac{1}{2}\sum_{\tiny\begin{matrix}U\leq\Fqn\\\dimq{U}=u\end{matrix}}\left(\qbin{\dimq{\mathcal{C}(U)}}{i} +\qbin{\dimq{\mathcal{C}[U]}}{i} \right)\\
	&=\frac{1}{2}\sum_{\tiny\begin{matrix}U\leq\Fqn\\\dimq{U}=u\end{matrix}}\left(\qbin{k-m(n-u)}{i} +\qbin{k-m(n-u)}{i} \right)\\
	&=\qbin{k-m(n-u)}{i}\sum_{\tiny\begin{matrix}U\leq\Fqn\\\dimq{U}=u\end{matrix}}1\\
	&=\qbin{n}{u}\qbin{k-m(n-u)}{i}.
\end{aligned}	
\end{equation*}

Suppose now $u<d_i$ and let $U\leq\Fqn$ be of dimension $u$. Assume, towards a contradiction, that $B_u^{(i)}(\mathcal{C})\ne0$, then there exists a subspace $\mathcal{D}$ of $\mathcal{C}(U)$ (or $\mathcal{C}[U]$ respectively) of dimension~$i$. We have  $d_i\leq\maxrank(\mathcal{D})\leq\dimq{U}=u<d_i$, which is a contradiction.
\end{proof}

\begin{remark}
\label{RemarkBuiAnt}
The connection between Definition~\ref{DefinitionBu} and optimal anticodes is the following.
	For $n<m$ and $0\leq u\leq n$, Lemma \ref{LemmaAntIntersection} implies
	\begin{equation*}
		B_u^{(i)}(\mathcal{C})=\displaystyle\sum_{A\in\mathcal{A}_u}\qbin{\dimq{\mathcal{C}\cap A}}{i}.
	\end{equation*}
	On the other hand, if $n=m$ and $1\leq i\leq k$ one can check that for $0\leq u\leq n-1$ we have
	\begin{equation*}
		B_u^{(i)}(\mathcal{C})=\frac{1}{2}\sum_{A\in\mathcal{A}_u}\qbin{\dimq{\mathcal{C}\cap A}}{i}.
	\end{equation*}
	Finally, in the case $u=n=m$ and $1\leq i\leq k$ we have 
	\begin{equation*}
		B_u^{(i)}(\mathcal{C})=\frac{1}{2}\left(B_{\Fqn}^{(i)}(\mathcal{C},c)+B_{\Fqn}^{(i)}(\mathcal{C},r)\right)=\qbin{\dimq{\mathcal{C}}}{i},
	\end{equation*}
	while in the case $n=m$ and $i=0$ we have
	\begin{equation*}
		B_u^{(i)}(\mathcal{C})=\frac{1}{2}\sum_{\tiny\begin{matrix}U\leq\Fqn\\\dimq{U}=u\end{matrix}} \left(\qbin{\dimq{\mathcal{C}(U)}}{0}+\qbin{\dimq{\mathcal{C}[U]}}{0}\right)=\qbin{n}{u}.
	\end{equation*}
\end{remark}

We now introduce a new invariant of $\mathcal{C}$ which extends the notion of rank distribution, showing that this invariant encodes the same information as the generalized binomial moments.

\begin{definition}
	The \textbf{$i$-th generalized rank weight distribution} of $\mathcal{C}$ is the integer vector whose $w$-th component, $0\leq w\leq n$, is defined by
	\begin{equation*}
		 A_w^{(i)}(\mathcal{C}):=\begin{cases}
			\displaystyle\sum_{\tiny\begin{matrix}W\leq\Fqn\\\dimq{W}=w\end{matrix}} A_{W}^{(i)}(\mathcal{C},c) &\; \textup{ if }\; n<m,\\\\
			\displaystyle\frac{1}{2}\sum_{\tiny\begin{matrix}W\leq\Fqn\\\dimq{W}=w\end{matrix}} \left(A_{W}^{(i)}(\mathcal{C},c)+ A_{W}^{(i)}(\mathcal{C},r)\right) &\; \textup{ if }\; n=m,
		\end{cases}
	\end{equation*}
	where, for every $W\leq\Fqn$,
	\begin{equation*}
	\begin{aligned}
		A_{W}^{(i)}(\mathcal{C},c)&:=
			|\{\mathcal{D}\leq\mathcal{C}:\dimq{\mathcal{D}}=i, \, \colsupp(\mathcal{D})=W\}|,\\
		A_{W}^{(i)}(\mathcal{C},r)&:=|\{\mathcal{D}\leq\mathcal{C}:\dimq{\mathcal{D}}=i, \, \rowsupp(\mathcal{D})=W\}|	.
	\end{aligned}
	\end{equation*}
\end{definition}

\begin{remark}
\label{RemarkAwne0}
	We have $A_w^{(i)}(\mathcal{C})=0$ for $0\leq w< d_i$. Indeed, suppose towards a contradiction that $A_w^{(i)}(\mathcal{C})\ne 0$ for a $0\leq w< d_i$. Then there exists $\mathcal{D}\leq\mathcal{C}$ of dimension $i$ such that
	\begin{equation*}
		\begin{cases}
			\dimq{\colsupp(\mathcal{D})}=w &\quad\textup{ if }\; n<m,\\
			\dimq{\colsupp(\mathcal{D})}=w \;\textup{ or }\; \dimq{\rowsupp(\mathcal{D})}=w &\quad\textup{ if }\; n=m.
		\end{cases}
	\end{equation*}
	This implies, by Theorem \ref{TheoremAnticodeEquiv}, that there exists an optimal anticode $A\in\mathcal{A}_w\subseteq \mathcal{A}$ such that~$\mathcal{D}\leq\mathcal{C}\cap A$, which contradicts the minimality of $d_i$.
\end{remark}

The following is the main theorem of this section. It gives inversion formulae connecting the $i$-th generalized rank weight distribution and the sequence of generalized binomial moments indexed by $i$. This result generalizes \cite[Lemma~30]{ravagnani2016rank}.

\begin{theorem}
\label{TheoremBuAw}
	The following hold for $0\leq u,w\leq n$.
	\begin{align}
	\begin{split}\label{EqBuAw}
		B_u^{(i)}(\mathcal{C})&= \sum_{w=0}^u \qbin{n-w}{u-w} A_w^{(i)}(\mathcal{C}),
	\end{split}\\
	\begin{split}\label{EqAwBu}
		A_w^{(i)}(\mathcal{C})&=\sum_{u=0}^w\qbin{n-u}{w-u}(-1)^{w-u}q^{\binom{w-u}{2}}B_u^{(i)}(\mathcal{C}).
	\end{split}
	\end{align}
\end{theorem}
\begin{proof}
First note that, for $U\leq\Fqn$,
\begin{eqnarray}
\label{EqBUAW1}
\sum_{W\leq U}A_W^{(i)}(\mathcal{C},c)&=&|\{\mathcal{D}\leq C:\dimq{\mathcal{D}}=i,\colsupp(\mathcal{D})\leq U\}|=B_{U}^{(i)}(\mathcal{C},c),\\
		\sum_{W\leq U}A_W^{(i)}(\mathcal{C},r)&=&|\{\mathcal{D}\leq C:\dimq{\mathcal{D}}=i,\rowsupp(\mathcal{D})\leq U\}|=B_{U}^{(i)}(\mathcal{C},r).
\label{EqBUAW2}
\end{eqnarray}
Therefore using the M\"obius inversion formula~\cite[Proposition 3.7.1]{stanleyec} we obtain
\begin{eqnarray}
	\label{EqAUBW1}
			A_W^{(i)}(\mathcal{C},c) &=& \sum_{U\leq W}(-1)^{\dimq{W}-\dimq{U}}q^{\binom{\dimq{W}-\dimq{U}}{2}}B_{U}^{(i)}(\mathcal{C},c),\\
		A_W^{(i)}(\mathcal{C},r) &=& \sum_{U\leq W}(-1)^{\dimq{W}-\dimq{U}}q^{\binom{\dimq{W}-\dimq{U}}{2}}B_{U}^{(i)}(\mathcal{C},r).
	\label{EqAUBW2}
\end{eqnarray}

	We continue the proof assuming $m=n$. The case $n<m$ is analogous and in fact simpler. By~\eqref{EqBUAW1} and~\eqref{EqBUAW2} we have
	\begin{eqnarray*}
			B_u^{(i)}(\mathcal{C})&=&\sum_{\tiny\begin{matrix}U\leq\Fqn\\\dimq{U}=u\end{matrix}} \left(\frac{1}{2}B_{U}^{(i)}(\mathcal{C},c)+\frac{1}{2} B_{U}^{(i)}(\mathcal{C},r)\right)\\
			&=&\sum_{\tiny\begin{matrix}U\leq\Fqn\\\dimq{U}=u\end{matrix}}\sum_{W\leq U}\left(\frac{1}{2}A_{W}^{(i)}(\mathcal{C},c)+\frac{1}{2} A_{W}^{(i)}(\mathcal{C},r)\right)\\
			&=&\sum_{w=0}^u\sum_{\tiny\begin{matrix}W\leq\Fqn\\\dimq{W}=w\end{matrix}}\left(\frac{1}{2}A_{W}^{(i)}(\mathcal{C},c)+\frac{1}{2} A_{W}^{(i)}(\mathcal{C},r)\right)\sum_{\tiny\begin{matrix}W\leq U\\\dimq{U}=u\end{matrix}}1\\
			&=&\sum_{w=0}^u\sum_{\tiny\begin{matrix}W\leq\Fqn\\\dimq{W}=w\end{matrix}}\left(\frac{1}{2}A_{W}^{(i)}(\mathcal{C},c)+\frac{1}{2} A_{W}^{(i)}(\mathcal{C},r)\right)\qbin{n-w}{u-w}\\
			&=&\sum_{w=0}^u\qbin{n-w}{u-w}\sum_{\tiny\begin{matrix}W\leq\Fqn\\\dimq{W}=w\end{matrix}}\left(\frac{1}{2}A_{W}^{(i)}(\mathcal{C},c)+\frac{1}{2} A_{W}^{(i)}(\mathcal{C},r)\right)\\
			&=&\sum_{w=0}^u \qbin{n-w}{u-w} A_w^{(i)}(\mathcal{C}).
		\end{eqnarray*}
	On the other hand, by~\eqref{EqAUBW1} and~\eqref{EqAUBW2} we get
	\begin{eqnarray*}
			A_w^{(i)}(\mathcal{C})&=&
		\sum_{\tiny\begin{matrix}W\leq\Fqn\\\dimq{W}=w\end{matrix}} \left(\frac{1}{2}A_{W}^{(i)}(\mathcal{C},c)+\frac{1}{2} A_{W}^{(i)}(\mathcal{C},r)\right)\\
		&=&\sum_
		{\tiny\begin{matrix}
			W\leq\Fqn\\
			\dimq{W}=w
		\end{matrix}}\sum_{u=0}^w
		\sum_
		{\tiny\begin{matrix}
			U\leq W\\
			\dimq{U}=u
		\end{matrix}}
		(-1)^{w-u}q^{\binom{w-u}{2}}\left(\frac{1}{2}B_{U}^{(i)}(\mathcal{C},c)+\frac{1}{2} B_{U}^{(i)}(\mathcal{C},r)\right)\\
		&=&\sum_{u=0}^w(-1)^{w-u}q^{\binom{w-u}{2}}
		\sum_
		{\tiny\begin{matrix}
			U\leq\Fqn\\
			\dimq{U}=u
		\end{matrix}}\left(\frac{1}{2}B_{U}^{(i)}(\mathcal{C},c)+\frac{1}{2} B_{U}^{(i)}(\mathcal{C},r)\right)
		\sum_{\tiny\begin{matrix}
			U\leq W\\
			\dimq{W}=w
		\end{matrix}}1\\
		&=&\sum_{u=0}^w(-1)^{w-u}q^{\binom{w-u}{2}}
		\sum_
		{\tiny\begin{matrix}
			U\leq\Fqn\\
			\dimq{U}=u
		\end{matrix}}\left(\frac{1}{2}B_{U}^{(i)}(\mathcal{C},c)+\frac{1}{2} B_{U}^{(i)}(\mathcal{C},r)\right)\qbin{n-u}{w-u}\\
		&=&\sum_{u=0}^w(-1)^{w-u}q^{\binom{w-u}{2}}\qbin{n-u}{w-u}
		\sum_
		{\tiny\begin{matrix}
			U\leq\Fqn\\
			\dimq{U}=u
		\end{matrix}}\left(\frac{1}{2}B_{U}^{(i)}(\mathcal{C},c)+\frac{1}{2} B_{U}^{(i)}(\mathcal{C},r)\right)\\
		&=&\sum_{u=0}^w(-1)^{w-u}q^{\binom{w-u}{2}}\qbin{n-u}{w-u}B_u^{(i)}(\mathcal{C}).
	\end{eqnarray*}
	This concludes the proof.
\end{proof}

The following generalizes the definition of rank weight enumerator. 

\begin{definition}
\label{DefinitionRankEnum}
	The \textbf{$i$-th generalized rank weight enumerator} of $\mathcal{C}$ is the homogeneous polynomial of degree $n$ in $\mathbb{Q}[X,Y]$ defined by
	\begin{equation*}
		W_\mathcal{C}^{(i)}(X,Y):=\sum_{w=0}^nA_w^{(i)}(\mathcal{C})X^{n-w}Y^w=\sum_{w=d_i}^nA_w^{(i)}(\mathcal{C})X^{n-w}Y^w.
	\end{equation*}
\end{definition}

The coefficients of $W_\mathcal{C}^{(i)}(X,Y)$ with respect to the basis $\{X^sY^t:s,t\in\mathbb{Z}_{\geq0}\}$ are the $A_w^{(i)}(\mathcal{C})$. Another well known $\mathbb{Q}$-basis of the ring of homogeneous polynomials of degree $n$ is given by~$\{\mathcal{B}_{n,u}(X,Y;q): n,u\in\mathbb{Z}_{\geq 0}, u\leq n\}$, where $\mathcal{B}_{n,u}(X,Y;q)$ is the \textbf{${(n,u)}$-th $q$-Bernstein polynomial}~\cite{lupas1987q} defined by
\begin{equation*}
	\mathcal{B}_{n,u}(X,Y;q):=\qbin{n}{u}Y^u\prod_{j=0}^{n-u-1}(X-q^jY).
\end{equation*}
The inversion formula associated with these polynomials is
\begin{equation*}
	X^{n-t}Y^t=\qbin{n}{t}^{-1}\sum_{u=t}^n\qbin{u}{t}\mathcal{B}_{n,u}(X,Y;q).
\end{equation*}

In the reminder of the section we compute the coefficients of 
$W_\mathcal{C}^{(i)}(X,Y)$ with respect to the basis given by the 
$q$-Bernstein polynomials.

\begin{definition}
	For every $0\leq u\leq n$, the \textbf{$(u,i)$-th normalized generalized binomial moment} is
	\begin{equation*}
		b_u^{(i)}(\mathcal{C}):=\begin{cases}
			0 &\;\textup{ if }\; u<0,\\
			\displaystyle\frac{B_{u+d_i}^{(i)}}{\qbin{n}{u+d_i}}&\;\textup{ if }\; 0\leq u \leq n-d^\perp-d_i,\\
			\qbin{k-m(n-u-d_i)}{i}  &\;\textup{ if }\; u>n-d^\perp-d_i.
		\end{cases}
	\end{equation*}
\end{definition}

\begin{theorem}
\label{TheoremWBern}
	We have
	\begin{equation*}
		W_{\mathcal{C}}^{(i)}(X,Y)=\sum_{u=d_i}^nb_{u-d_i}^{(i)}(\mathcal{C})\mathcal{B}_{n,u}(X,Y;q).
	\end{equation*}
\end{theorem}

\begin{proof}
	Applying the inversion formula for the $q$-Bernstein polynomial to Definition~\ref{DefinitionRankEnum} we obtain
	\begin{equation*}
		\begin{aligned}
			W_{\mathcal{C}}^{(i)}(X,Y)&=\sum_{w=d_i}^nA_w^{(i)}(\mathcal{C})\qbin{n}{w}^{-1}\sum_{u=w}^n\qbin{u}{w}\mathcal{B}_{n,u}(X,Y;q)\\
			&=\sum_{u=d_i}^n\mathcal{B}_{n,u}(X,Y;q)\sum_{w=d_i}^u\qbin{n}{w}^{-1}\qbin{u}{w}A_w^{(i)}(\mathcal{C})\\
			&=\sum_{u=d_i}^n\mathcal{B}_{n,u}(X,Y;q)\qbin{n}{u}^{-1}\sum_{w=d_i}^u\qbin{n-w}{u-w}A_w^{(i)}(\mathcal{C})\qquad\qquad \textup{(by Lemma \ref{LemmaTools})}
			\\
			&=\sum_{u=d_i}^n\mathcal{B}_{n,u}(X,Y;q)b_{u-d_i}^{(i)}(\mathcal{C}),
		\end{aligned}
	\end{equation*}
	where the latter inequality follows by Theorem \ref{TheoremBuAw} and the definition of normalized generalized binomial moments.
\end{proof}

\section{$i$-BMD rank-metric codes}
\label{SectioniMRD}

In this section we introduce a new family of extremal rank-metric codes, namely the $i$-BMD codes, whose generalized rank weight distributions and generalized binomials moments depend only on the code parameters $n,m,k,d$. Introducing this family of codes is also motivated by the fact that it refines the notions of MRD and DQMRD. Ducoat and Oggier \cite[Definition~2]{ducoat2014rank}, in the $\Fqm$-linear case, and de la Cruz \cite[Definition~5.1]{de2018dually}, for some values of the parameter $i$, defined the family of codes that are optimal with respect to the $i$-th generalized rank weight. We extend these definitions referred to as $i$-MRD codes and we show that this latter family properly contains the class of $i$-BMD codes. For the remainder $\alpha,\rho$ will denote non-negative integers such that $k=\alpha m+\rho$ and $0\leq \rho\leq m-1$.

\begin{definition}
	The code $\mathcal{C}$ is $i$-\textbf{BMD} (\textbf{Binomial Moment Determined}) if  
	\begin{equation*}
		n-d^\perp-d_i<0.
	\end{equation*}
\end{definition}

\begin{lemma}
\label{LemmajMRD}
	If $\mathcal{C}$ is $i$-BMD then $\mathcal{C}$ is $j$-BMD for all $i\leq j\leq k$. Moreover, every non-zero matrix rank-metric code is $i$-BMD for all $i$ such that $ \left\lceil\frac{i}{m}\right\rceil\geq n$.
\end{lemma}
\begin{proof}
	The first part of the statement is an immediate consequence of the fact that the generalized rank weights form an increasing sequence. 
	For the second part, Lemma \ref{LemmaGenRankWeights} implies that~$d_i+d^\perp\geq\left\lceil\frac{i}{m}\right\rceil+1$, from which the statement follows.
\end{proof}

\begin{definition}
	We say that $\mathcal{C}$ is \textbf{minimally $i$-BMD} if $i=\min\{j: \mathcal{C} \textup{ is }j\textup{-BMD}\}$.
\end{definition}

Notice that the codes $\{0\}$ and $\Fqnm$ are minimally $1$-BMD. A useful property of an $i$-BMD code is that, for any $i\leq j\leq k$, its $j$-th generalized rank weight distributions and binomial moments depend only on the code parameters $n$, $m$, $k$ and $d_i$. We give explicit formulae in the following results.

\begin{lemma}
\label{LemmaBuBMD}
	Let $\mathcal{C}$ be $i$-BMD. For all $i\leq j\leq k$ and $0\leq u\leq n$ we have
	\begin{equation*}
		B_u^{(j)}(\mathcal{C})=
		\begin{cases}
			0 & \textup{ if }\; u<d_j,\\
			\qbin{n}{u}
		\qbin{k-m(n-u)}{j} & \textup{ if }\; u\geq d_j.
		\end{cases}
	\end{equation*}
\end{lemma}
\begin{proof}
	It follows from Lemmas \ref{LemmaBu} and \ref{LemmajMRD} and the Definition of a $j$-BMD code.
\end{proof}

As an immediate consequence we have the following.

\begin{lemma}
\label{LemmabuBMD}
	Let $\mathcal{C}$ be $i$-BMD. For all $i\leq j\leq k$  we have
	\begin{enumerate}
		\item $\displaystyle b_u^{(j)}(\mathcal{C})=\begin{cases}
			0 &\;\textup{ if }\; u<0,\\
			\qbin{k-m(n-u-d_j)}{j}  &\;\textup{ if }\; u\geq 0;
		\end{cases}$
		\item $\displaystyle A_w^{(j)}(\mathcal{C})=\qbin{n}{w}\sum_{u=d_j}^w\qbin{w}{u}(-1)^{w-u}q^{\binom{w-u}{2}}\qbin{k-m(n-u)}{j}\quad$ for all $0\leq w\leq n$;
		\item $\displaystyle W_{\mathcal{C}}^{(j)}(X,Y)=\sum_{u=d_i}^n\qbin{k-m(n-u)}{j}\mathcal{B}_{n,u}(X,Y;q)$.
	\end{enumerate}
\end{lemma}

We now introduce another family of extremal codes, extending the definitions in~\cite{ducoat2014rank,de2018dually}.

\begin{definition}
	The code $\mathcal{C}$ is \textbf{$i$-MRD} if 
	\begin{equation*}
		d_i=n-\left\lfloor\frac{k-i}{m}\right\rfloor.
	\end{equation*}
\end{definition}

\begin{remark}
	According to the definition above, the codes in the family introduced in~\cite{ducoat2014rank} are $(rm)$-MRD codes, for $r\in\left\{1,2,\ldots,\frac{k}{m}\right\}$ when $m|k$, while those in~\cite{de2018dually} are~$i$-MRD codes with $i=1+(r-1)m$ and $r\in\left\{1,2,\ldots,\left\lceil\frac{k}{m}\right\rceil\right\}$.
\end{remark}

The following result extends \cite[Lemma~6.2]{de2018dually}.

\begin{lemma}
	If $\mathcal{C}$ is $i$-MRD then $\mathcal{C}$ is $(i+m)$-MRD.
\end{lemma}
\begin{proof}
	Let $\mathcal{C}$ be an $i$-MRD code. We have
	\begin{equation*}
		d_i=n-\left\lfloor\frac{k-i}{m}\right\rfloor
	\end{equation*}
	and Lemma \ref{LemmaGenRankWeights} implies
	\begin{equation*}
		n-\left\lfloor\frac{k-i}{m}\right\rfloor+1=d_i+1\leq d_{i+m}\leq n-\left\lfloor\frac{k-i-m}{m}\right\rfloor=n-\left\lfloor\frac{k-i}{m}\right\rfloor+1.\qedhere
	\end{equation*}
\end{proof}

In the following example we show that if $\mathcal{C}$ is $i$-MRD then it is not necessarily $(i+1)$-MRD. 
In the remainder of the paper, we denote the codes in the examples by $\mathcal{C}_v$ and use them repeatedly. 

\begin{example}
\label{Example(i)not(i+1)MRD}
	Let $\mathcal{C}_1\leq \mathbb{F}_2^{3\times 4}$ be the code of dimension $6$ generated by
	\begin{equation*}
	\scriptsize\left\{
		\begin{pmatrix}
        1 & 0 & 0 & 0\\
        0 & 0 & 0 & 1\\
        1 & 0 & 1 & 1
        \end{pmatrix},
        \begin{pmatrix}
        0 & 1 & 0 & 0\\
        1 & 0 & 0 & 1\\
        0 & 1 & 1 & 0
        \end{pmatrix},
        \begin{pmatrix}
        0 & 0 & 1 & 0\\
        1 & 0 & 0 & 1\\
        1 & 0 & 1 & 0
        \end{pmatrix},\\
        \begin{pmatrix}
        0 & 0 & 0 & 1\\
        1 & 0 & 0 & 1\\
        1 & 1 & 0 & 0
        \end{pmatrix},
        \begin{pmatrix}
        0 & 0 & 0 & 0\\
        0 & 1 & 0 & 1\\
        0 & 0 & 0 & 1
        \end{pmatrix},
        \begin{pmatrix}
        0 & 0 & 0 & 0\\
        0 & 0 & 1 & 1\\
        1 & 1 & 0 & 1
        \end{pmatrix}\right\}.
	\end{equation*}
	One can check that the generalized rank weights of $\mathcal{C}$ are
	\begin{equation*}
		d_1(\mathcal{C}_1)=1, \qquad d_2(\mathcal{C}_1)=d_3(\mathcal{C}_1)=2, \qquad d_4(\mathcal{C}_1)=d_5(\mathcal{C}_1)=d_6(\mathcal{C}_1)=3
	\end{equation*}
	and that
	\begin{equation*}
		3-\left\lfloor\frac{6-2}{4}\right\rfloor=2=d_2(\mathcal{C}_1) \qquad\qquad  3-\left\lfloor\frac{6-3}{4}\right\rfloor=3\neq 2=d_3(\mathcal{C}_1).
	\end{equation*}
	Therefore $\mathcal{C}_1$ is $2$-MRD but is not $3$-MRD.
\end{example}

\begin{lemma}[{\cite[Theorem~22]{de2018weight}}]
\label{LemmadiDQMRD}
	If $\mathcal{C}$ is DQMRD, then its generalized rank weights are determined by the parameters $n$, $m$ and $k$ as follows:
	\begin{equation*}
		d_i=
		\begin{cases}
			n-\alpha & 1\leq i\leq\rho,\\
			n+1+s-\alpha & \rho+1+sm\leq i\leq \rho+(s+1)m \quad 0\leq s\leq \alpha-2,\\
			n & \rho+1+(\alpha-1)m\leq i\leq k .
		\end{cases}
	\end{equation*}
\end{lemma}

\begin{lemma}
\label{CorollaryCompactDQMRD}
	The code $\mathcal{C}$ is DQMRD if and only if $m\nmid k$ and $\mathcal{C}$ is $i$-MRD for all $1\leq i\leq k$.
\end{lemma}
\begin{proof}
	Recall that $k=\alpha m+\rho$, with $0\leq \rho\leq m-1$. We have
	\begin{equation*}
		n-\left\lfloor\frac{k-i}{m}\right\rfloor=n-\left\lfloor\frac{\alpha m+\rho-i}{m}\right\rfloor=n-\alpha-\left\lfloor\frac{\rho-i}{m}\right\rfloor=n-\alpha+\left\lceil\frac{i-\rho}{m}\right\rceil.
	\end{equation*}
	Observe that
	\begin{itemize}
		\item $0\leq \frac{\rho-i}{m}<1$ if $1\leq i\leq \rho$, in which case $\left\lfloor\frac{\rho-i}{m}\right\rfloor=0$;
		\item $s+\frac{1}{m}\leq\frac{i-\rho}{m}\leq s+1$ if $\rho+1+sm\leq i\leq \rho+(s+1)m$, in which case $\left\lceil\frac{i-\rho}{m}\right\rceil=s+1$;
		\item $\alpha-1+\frac{1}{m}\leq \frac{i-\rho}{m}\leq \alpha$ if $\rho+1+(\alpha-1)m\leq i\leq k$, in which case $\left\lceil\frac{i-\rho}{m}\right\rceil=\alpha$.
	\end{itemize}
	The result now follows by direct application of Lemma \ref{LemmadiDQMRD}.
\end{proof}

We devote the remaining part of this section to showing that if a code is $i$-BMD then it is $i$-MRD and that, in general, the converse does not hold. Assume that $\mathcal{C}$ has dimension $k\in\{1,\ldots, nm-1\}$. For an integer $1\leq p\leq m$, we define the sets
\begin{equation*}
	\begin{aligned}
		V_p(\mathcal{C})&:=\left\{d_{p+jm}:1\leq p+jm\leq k\right\},\\
		\overline{V}_p(\mathcal{C})&:=\{n+1-x: x\in V_p(\mathcal{C})\}.
	\end{aligned}
\end{equation*}
 In \cite[Corollary~38]{ravagnani2016generalized}, it is shown that $V_p(\mathcal{C}^\perp)=\{1,2,\ldots,n\}\setminus \overline{V}_{p+k}(\mathcal{C})$. In particular, the generalized rank weights of $\mathcal{C}$ completely determine the generalized rank weights of $\mathcal{C}^\perp$.

For the remainder, we assume that $r,t$ are positive integers such that
\begin{equation*}
	i=k+r-tm \qquad \textup{ with }\quad 1\leq r\leq m.
\end{equation*}

It is easy to check that these integers exist for any $i$. Moreover, $1\leq i\leq k$ implies $1\leq t\leq \left\lfloor\frac{k+r-1}{m}\right\rfloor$.

\begin{lemma}
\label{Lemmadi=t}
	Let $2\leq i\leq k-1$ and write $i=k+r-tm$ with $1\leq t\leq \left\lfloor\frac{k+r-1}{m}\right\rfloor$ and $1\leq r\leq m$. If $\mathcal{C}$ is $i$-BMD and $r\leq k^\perp$ then $d_{k+r-jm}=n-j+1$ for all $1\leq j\leq t$.
\end{lemma}
\begin{proof}
	Let $\mathcal{C}$ be $i$-BMD, then $n+1-d_i\leq d^\perp\leq d_r^\perp$. We have
	\begin{equation*}
		\begin{aligned}
			V_r(\mathcal{C}^\perp):&=\left\{d_{r+jm}^\perp:0\leq j\leq\left\lfloor\frac{k^\perp-r}{m}\right\rfloor\right\},\\
			\overline{V}_{k+r}(\mathcal{C}):&=\left\{n+1-d_{k+r-jm}:1\leq j\leq \left\lfloor\frac{k+r-1}{m}\right\rfloor\right\}
		\end{aligned}
	\end{equation*}
		as $0\leq r-1\leq m-1$. Now,
		\begin{equation*}
		\begin{aligned}
			d_r^\perp&=\min(V_r(\mathcal{C}^\perp))=\min\left(\{1,2,\ldots,t\}\setminus\overline{V}_{k+r}(\mathcal{C})\right)>n+1-d_i,
		\end{aligned}
		\end{equation*}
		therefore in particular, $\{1,2,\ldots,t\}\subseteq\overline{V}_{k+r}(\mathcal{C})$.
\end{proof}

\begin{lemma}
\label{Lemmar>kperp}
	Let $i\geq 2$, $\mathcal{C}$ be minimally $i$-BMD and suppose $1\leq k^\perp<r$, $2\leq r\leq m$. Then exactly one of the following holds.
	\begin{enumerate}
	\item $i>k+1-m$, $d^\perp=1$ and $d_{k+1-jm}=n-j$ for all $1\leq j\leq \left\lfloor\frac{k}{m}\right\rfloor$.
	\item There exists an integer $1\leq s\leq \left\lfloor\frac{k}{m}\right\rfloor$ such that $k+1-(s+1)m< i\leq k+1-sm$, $d^\perp=s+1$ and
		\begin{equation*}
			\begin{cases}
				d_{k+1-jm}=n-j+1 & \textup{ if } 1\leq j\leq s,\\
				d_{k+1-jm}=n-j & \textup{ if } s+1\leq j\leq \left\lfloor\frac{k}{m}\right\rfloor.
			\end{cases}
		\end{equation*}
	\end{enumerate}
\end{lemma}
\begin{proof}
	Define the sets
	\begin{equation*}
		\begin{aligned}
			V_1(\mathcal{C}^\perp):&=\left\{d_{1+jm}^\perp:0\leq j\leq\left\lfloor\frac{k^\perp-1}{m}\right\rfloor\right\}=\{d^\perp\},\\
			\overline{V}_{k+1}(\mathcal{C}):&=\left\{n+1-d_{k+r-jm}:1\leq j\leq \left\lfloor\frac{k}{m}\right\rfloor\right\}.\\
		\end{aligned}
	\end{equation*}
	Suppose $i>k+1-m$. Since $\mathcal{C}$ is minimally $i$-BMD, we have $n+1-d_{k+1-jm}\geq d^\perp$ for all $1\leq j\leq \left\lfloor\frac{k}{m}\right\rfloor$. The fact that $V_1(\mathcal{C}^\perp)\sqcup\overline{V}_{k+1}(\mathcal{C})=\{1,\ldots,n\}$ implies $d^\perp=1$ and
	\begin{equation*}
		\begin{aligned}
			n+1-d_{k+1-m}&=2\\
			n+1-d_{k+1-2m}&=3\\ 
			\vdots\\
			n+1-d_{k+1-\left\lfloor\frac{k}{m}\right\rfloor m}&=\left\lfloor\frac{k}{m}\right\rfloor+1.
		\end{aligned}
	\end{equation*}
	This establishes (1). Suppose now that there exists an integer $1\leq s\leq \left\lfloor\frac{k}{m}\right\rfloor$ such that $k+1-(s+1)m< i\leq k+1-sm$ then, by Lemma \ref{Lemmadi=t}, we have $d_{k+1-jm}=n-j+1$ for $1\leq j\leq s$. On the other hand we have $n+1-d_{k+1-sm}\leq d^\perp \leq n+1-d_{k+1-(s+1)m}$ since $\mathcal{C}$ is minimally $i$-BMD. Therefore, the fact that $V_1(\mathcal{C}^\perp)\sqcup\overline{V}_{k+1}(\mathcal{C})=\{1,\ldots,n\}$ implies~$d^\perp=s+1$ and
		\begin{equation*}
		\begin{aligned}
			n+1-d_{k+1-(s+1)m}&=s+2\\
			n+1-d_{k+1-(s+2)m}&=s+3\\ 
			\vdots\\
			n+1-d_{k+1-\left\lfloor\frac{k}{m}\right\rfloor m}&=\left\lfloor\frac{k}{m}\right\rfloor+1.
		\end{aligned}
		\end{equation*}
	This concludes the proof.
\end{proof}

\begin{lemma}
\label{LemmakBMD}
	If $\mathcal{C}$ is $k$-BMD then $d_k=n$. 
\end{lemma}
\begin{proof}
	Suppose toward a contradiction that there exists a $k$-BMD code $\mathcal{C}$ with $d_k=\delta<n$. By definition of generalized rank weights we get
	\begin{align*}
			\delta &=\frac{1}{m}\min\{\dimq{A}:A\in\mathcal{A},\;\dimq{\mathcal{C}\cap A}\geq k\}\\
			&=\frac{1}{m}\min\{\dimq{A}:A\in\mathcal{A},\; \mathcal{C}\subseteq A\}
		\end{align*}
	and therefore $\mathcal{C}$ must be contained in an optimal anticode $A\in\mathcal{A}(n\times m,\delta)$. Let $B\in\gl_n(\Fq)$ be a matrix such that $$B\cdot\mathcal{C} \subseteq \Fqnm\left(\left<e_1,\ldots,e_\delta\right>_q\right).$$ Notice that in particular the last row of every matrix in $B\cdot\mathcal{C}$ must be of all zeros. Hence $(B\cdot\mathcal{C})^\perp$ must contains the matrix
	\begin{equation*}
		\begin{pmatrix}
			0 & 0 &\cdots &0\\
			\vdots & \vdots &&\vdots\\
			0 & 0 &\cdots &0\\
			1 & 0 &\cdots &0
		\end{pmatrix},
	\end{equation*}
	which implies $d^\perp=1$. Therefore, $n-d_k-d^\perp=n-\delta-1\geq 0$ and we get a contradiction. 
\end{proof}

We can now prove the main result of this section.

\begin{theorem}
\label{TheoremiBMD=>iMRD}
	If $\mathcal{C}$ is $i$-BMD then $\mathcal{C}$ is $i$-MRD.
\end{theorem}
\begin{proof}
	We already observed that the codes $\{0\}$ and $\Fqnm$ are $1$-MRD, it easy to check that that they are also $1$-BMD. Indeed,
	\begin{equation*}
			n-d(\{0\})-d(\Fqnm)=n-(n+1)-1=-2<0,
	\end{equation*} 
	so now assume that $\mathcal{C}$ is non-trivial. If $\mathcal{C}$ is $1$-BMD, then $n< d+d^\perp$. It follows that
	\begin{equation*}
		d+d^\perp=\begin{cases}
			n+1 & \textup{ if } m\nmid k,\\
			n+2 & \textup{ if } m|k,
		\end{cases}
	\end{equation*}
	which implies that $\mathcal{C}$ is either MRD or DQMRD and therefore $1$-MRD.
	Lemma \ref{LemmakBMD} implies~$d_k=n$ for a $k$-BMD code. Thus, $n=n-\left\lfloor\frac{k-k}{m}\right\rfloor$ and $\mathcal{C}$ is $k$-MRD.
	 
	Let $\mathcal{C}$ be $i$-BMD for an $2\leq i\leq k-1$. If $r\leq k^\perp$ then Lemma \ref{Lemmadi=t} implies $d_i=n-t+1$. Therefore,
			\begin{equation*}
				n+1-t=d_i\leq n-\left\lfloor\frac{k-i}{m}\right\rfloor=n-\left\lfloor\frac{k-k-r+tm}{m}\right\rfloor=n-t+1
			\end{equation*}
			which implies that $\mathcal{C}$ is $i$-MRD.
			
	We now consider the case $r>k^\perp$, in which case $2\leq r\leq m$. Let $\mathcal{C}$ be minimally $j$-BMD for some $j\in\{2,\ldots,i\}$. We consider two cases, $j>k+1-m$ and $j\leq k+1-m$.
	
	In the first case, namely $j>k+1-m$, Lemma \ref{Lemmar>kperp} implies $d^\perp=1$ and $d_{k+1-m}=n-1$. Therefore,
		\begin{equation*}
			n-1=d_{k+1-m}\leq d_{j-1}<d_j\leq d_i\leq d_k=n
		\end{equation*}
		which implies that $\mathcal{C}$ is $j$-MRD and $i$-MRD. Suppose now $j\leq k+1-m$, then there exists an integer $1\leq s\leq \left\lfloor\frac{k}{m}\right\rfloor$ such that $k+1-(s+1)m< j\leq k+1-sm$. Lemma \ref{Lemmar>kperp} implies that~$\mathcal{C}$ is $j$-MRD. Indeed, if $j=k+1-sm$ then
		\begin{equation*}
			n-s+1=d_{k+1-sm}\leq n-\left\lfloor\frac{k-k-1+sm}{m}\right\rfloor=n-s+1,
		\end{equation*}
		while if $k+1-(s+1)m<j< k+1-sm$ then
		\begin{equation*}
			n-s-1=d_{k+1-(s+1)m}\leq d_{j-1}<d_j\leq d_{k-sm}\leq  n-\left\lfloor\frac{k-k+sm}{m}\right\rfloor=n-s.
		\end{equation*}
		We have shown that if $r>k^\perp$ and $\mathcal{C}$ is minimally $j$-BMD then $\mathcal{C}$ is $j$-MRD. It remains to show that $\mathcal{C}$ is also $i$-MRD in this case. Recall that $k+2-tm\leq i\leq k-(t-1)m$ and $j\leq i$. Now, in the case $k+2-tm\leq j\leq i\leq k-(t-1)m$ we have
		\begin{equation*}
			n-t+1=d_j\leq d_i\leq d_{k-(t-1)m}\leq n-\left\lfloor\frac{k-k+(t-1)m}{m}\right\rfloor=n-t+1.
		\end{equation*}
		In the case $j<k+2-tm \leq i\leq k-(t-1)m$ we have 
		\begin{equation*}
			n-t+1=d_{k+1-tm}\leq d_i\leq d_{k-(t-1)m}\leq  n-\left\lfloor\frac{k-k+(t-1)m}{m}\right\rfloor=n-t+1.
		\end{equation*}
		Finally, if $j< k+2-m\leq i$, Lemma \ref{LemmakBMD} implies $n=d_{k+1-m}\leq d_i\leq d_k=n$.
\end{proof}

The following example shows that the converse of Theorem \ref{TheoremiBMD=>iMRD} does not hold.

\begin{example}
	Let $\mathcal{C}_1$ be the code of Example \ref{Example(i)not(i+1)MRD}. It is not hard to check that $\mathcal{C}_1$ is $2$-MRD but not $1$-MRD. Indeed, 
	\begin{equation*}
			 3-\left\lfloor\frac{8-2}{4}\right\rfloor=2=d_2(\mathcal{C}_1) \quad\textup{ but }\quad 3-\left\lfloor\frac{8-1}{4}\right\rfloor=2\ne 1 =d_1(\mathcal{C}_1).
	\end{equation*}
	We want to show that $\mathcal{C}_1$ is not $2$-BMD, i.e., that $n-d_2(\mathcal{C}_1)-d(\mathcal{C}_1^\perp)\geq 0$.
	
	Define the sets
	\begin{align*}
			V_1(\mathcal{C}_1^\perp):&=\{d_{1+4j}(\mathcal{C}_1^\perp): 1\leq 1+4j\leq \dimq{\mathcal{C}_1^\perp}\}\\
			&=\{d_{1+4j}(\mathcal{C}_1^\perp): 1\leq 1+4j\leq 6\}\\
			&=\{d_1(\mathcal{C}_1^\perp),d_5(\mathcal{C}_1^\perp)\},\\
			\\
			\overline{V}_{\dimq{\mathcal{C}_1}+1}(\mathcal{C}_1)=\overline{V}_7(\mathcal{C}_1):&=\{3+1-d_{7+4j}(\mathcal{C}_1): 1\leq 1+4j\leq \dimq{\mathcal{C}_1}\}\\
			&=\{4-d_{7+4j}(\mathcal{C}_1): 1\leq 1+4j\leq 6\}\\
			&=\{4-d_3(\mathcal{C}_1)\}=\{2\}.
		\end{align*}
	We have
	\begin{equation*}
		V_1(\mathcal{C}_1^\perp)=\{1,2,3\}\setminus\overline{V}_7(\mathcal{C}_1)=\{1,2,3\}\setminus\{2\}=\{1,3\}
	\end{equation*}
	that implies $d(\mathcal{C}_1^\perp)=1$. Therefore, $\mathcal{C}_1$ is not $2$-BMD because
	\begin{equation*}
		n-d_2(\mathcal{C}_1)-d_1(\mathcal{C}_1^\perp)=3-2-1=0\geq 0.
	\end{equation*}
\end{example}

\begin{remark}
	The classes of MRD and QMRD codes partition the set of $1$-MRD codes. Theorem \ref{TheoremiBMD=>iMRD} implies that the set of $1$-BMD codes is contained in the set of $1$-MRD codes. Observe that the only codes that are $1$-MRD but not $1$-BMD are those that are QMRD but not DQMRD. Figure \ref{Fig1MRD} shows the set representation of the class of $1$-MRD codes.
	
\begin{figure}[h]

\centering
\tikzset{every picture/.style={line width=0.75pt}} %set default line width to 0.75pt        

\begin{tikzpicture}[x=0.75pt,y=0.75pt,yscale=-1,scale=0.8]
%uncomment if require: \path (0,559); %set diagram left start at 0, and has height of 559

%Shape: Ellipse [id:dp24188024432234023] 
\draw   (319.82,260.13) .. controls (380.37,260.09) and (429.47,298.11) .. (429.51,345.05) .. controls (429.54,392) and (380.49,430.09) .. (319.94,430.13) .. controls (259.4,430.17) and (210.29,392.15) .. (210.26,345.2) .. controls (210.22,298.26) and (259.28,260.17) .. (319.82,260.13) -- cycle ;
%Shape: Ellipse [id:dp7760255396062348] 
\draw   (221.2,258.2) .. controls (281.74,258.16) and (330.85,296.18) .. (330.88,343.13) .. controls (330.91,390.07) and (281.86,428.16) .. (221.32,428.2) .. controls (160.77,428.25) and (111.66,390.23) .. (111.63,343.28) .. controls (111.6,296.34) and (160.65,258.25) .. (221.2,258.2) -- cycle ;

% Text Node
\draw  [color={rgb, 255:red, 255; green, 255; blue, 255 }  ,draw opacity=1 ][fill={rgb, 255:red, 255; green, 255; blue, 255 }  ,fill opacity=1 ]  (165.3,250.99) -- (227.3,250.99) -- (227.3,272.99) -- (165.3,272.99) -- cycle  ;
\draw (196.3,261.99) node  [align=left] {QMRD};

% Text Node
\draw  [color={rgb, 255:red, 255; green, 255; blue, 255 }  ,draw opacity=1 ][fill={rgb, 255:red, 255; green, 255; blue, 255 }  ,fill opacity=1 ]  (315.3,250.99) -- (373.3,250.99) -- (373.3,272.99) -- (315.3,272.99) -- cycle  ;
\draw (345.3,261.99) node  [align=left] {$1$-BMD};

% Text Node
\draw (378.39,345.49) node  [align=left] {MRD};

% Text Node
\draw (271.39,345.49) node  [align=left] {DQMRD};
\end{tikzpicture}
\caption{$1$-MRD codes.}\label{Fig1MRD}
\end{figure}
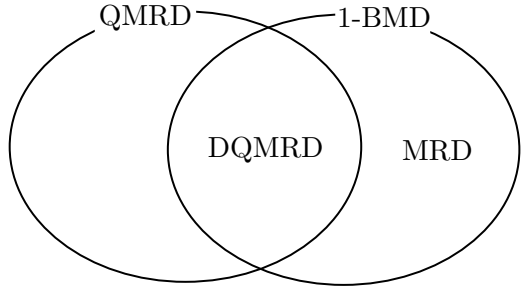

\end{remark}

\begin{theorem}
	Let $i\geq 2$. If $\mathcal{C}$ is minimally $i$-BMD then $\mathcal{C}$ is not $(i-1)$-MRD.
\end{theorem}
\begin{proof}
	Suppose $i=2$ and that $\mathcal{C}$ is $1$-MRD and minimally $2$-BMD code $\mathcal{C}$. We have
	\begin{equation*}
		d^\perp>n-d_2=n-n+\left\lfloor\frac{k-2}{m}\right\rfloor=\left\lfloor\frac{k-2}{m}\right\rfloor,
	\end{equation*}
	and
	\begin{equation*}
		n-d-d^\perp=n-n+\left\lfloor\frac{k-1}{m}\right\rfloor-d^\perp=\left\lfloor\frac{k-1}{m}\right\rfloor-d^\perp<\left\lfloor\frac{k-1}{m}\right\rfloor-\left\lfloor\frac{k-2}{m}\right\rfloor\leq 0
	\end{equation*} 
	which implies that $\mathcal{C}$ is $1$-BMD, yielding a contradiction. Now assume $i\geq 3$. Since $\mathcal{C}$ is $(i-1)$-MRD and minimally $i$-BMD we have 
	\begin{equation}
	\label{Eq(i-1)MRD(i)BMD_2}
		\left\lfloor\frac{k-i}{m}\right\rfloor=n-d_i<d^\perp\leq n-d_{i-1}=\left\lfloor\frac{k-i+1}{m}\right\rfloor. 
	\end{equation}
	Recall that $k=\alpha m+\rho$. Write $i=\beta m+\sigma$ with $0\leq\sigma\leq m-1$ and $\beta\leq \alpha$. Therefore,~\eqref{Eq(i-1)MRD(i)BMD_2} can be rewritten as
	\begin{equation}
	\label{Eq(i-1)MRD(i)BMD}
		\alpha-\beta+\left\lfloor\frac{\rho-\sigma}{m}\right\rfloor<d^\perp\leq \alpha-\beta+\left\lfloor\frac{\rho-\sigma+1}{m}\right\rfloor.
	\end{equation}
	Table~\ref{TableiBMDnot(i-1)MRD} below shows the values of $\rho$, $\beta$ and $\sigma$ for which the relation in~\eqref{Eq(i-1)MRD(i)BMD} holds. Notice that $i\geq 3$ and $\beta=0$ imply $\sigma\geq 3$.

\setlength{\extrarowheight}{5pt}
\renewcommand\tabularxcolumn[1]{m{#1}}
\begin{tabularx}{\linewidth}{|>{\hsize=0.15\hsize}Y|>{\hsize=0.3\hsize}c|>{\hsize=0.35\hsize}c|>{\hsize=0.2\hsize}c|}
	\hline
			\multirow{3}{*}{$\rho=0$, $\beta\geq 1$} & $\sigma=0$ & $\alpha-\beta<d^\perp\leq \alpha-\beta$ & $\Rightarrow\!\Leftarrow$\\\cline{2-4}
			& $\sigma=1$ & $\alpha-\beta-1<d^\perp\leq \alpha-\beta$ & $d^\perp=\alpha-\beta$\\\cline{2-4}
			& $2\leq \sigma\leq m-1$ & $\alpha-\beta-1<d^\perp\leq \alpha-\beta-1$& $\Rightarrow\!\Leftarrow$\\\hline
			\multirow{5}{*}{$\rho\geq 1$, $\beta\geq 1$} & $\sigma=0$, $\rho\leq m-2$ & $\alpha-\beta<d^\perp\leq \alpha-\beta$ & $\Rightarrow\!\Leftarrow$\\\cline{2-4}
			& $\sigma=0$, $\rho=m-1$ & $\alpha-\beta<d^\perp\leq \alpha-\beta+1$ & $d^\perp=\alpha-\beta+1$\\\cline{2-4}
			& $\sigma\leq \rho$ & $\alpha-\beta<d^\perp\leq \alpha-\beta$ & $\Rightarrow\!\Leftarrow$\\\cline{2-4}
			& $\sigma=\rho+1$ &  $\alpha-\beta-1<d^\perp\leq \alpha-\beta$ & $d^\perp=\alpha-\beta$\\\cline{2-4}
			& $\rho+2\leq \sigma\leq m-1$ & $\alpha-\beta-1<d^\perp\leq \alpha-\beta-1$ & $\Rightarrow\!\Leftarrow$\\\hline
			\multirow{3}{*}{$\rho\geq 0$, $\beta=0$} & $3\leq\sigma\leq \rho$ & $\alpha<d^\perp\leq \alpha$ & $\Rightarrow\!\Leftarrow$\\\cline{2-4}
			& $\min\{3,\rho+2\}<\sigma\leq m-1$ & $\alpha-1<d^\perp\leq\alpha-1$ & $\Rightarrow\!\Leftarrow$\\\cline{2-4}
			& $\sigma=\rho+1$, $\rho\geq 2$ & $\alpha-1<d^\perp\leq\alpha$ & $d^\perp=\alpha$\\\hline
	\caption{\label{TableiBMDnot(i-1)MRD}}
\end{tabularx}
	
It remains to check that for each case in the Table~\ref{TableiBMDnot(i-1)MRD} for which $d^\perp$ is determined we get a contradiction. We show only this for the first of such cases in the table and omit the proofs for remaining cases. Assume $\rho=0$, $\beta\geq 1$ and $\sigma=1$ then we have $d^\perp =\alpha-\beta$ by Table \ref{TableiBMDnot(i-1)MRD}. On the other hand, $k=m\alpha$, $i=\beta m+1=\alpha m+1-(\alpha-\beta)m$ and we have 
\begin{equation*}
	\begin{aligned}
		V_1(\mathcal{C}^\perp)&=\{d_{1+jm}^\perp:0\leq j\leq n-\alpha-1\},\\
		\overline{V}_{\alpha m+1}&=\{n+1-d_{\alpha m+1-jm}(\mathcal{C}):1\leq j\leq \alpha\}.
	\end{aligned}
\end{equation*}
Since $\mathcal{C}$ is minimally $i$-BMD, Theorem~\ref{Lemmadi=t} implies $d_i=n-\alpha+\beta+1$ and therefore $d^\perp=\alpha-\beta+1$ which yields a contradiction.
\end{proof}

We give a graphical illustration of the families of $i$-MRD and $i$-BMD codes in Figure~\ref{FigInclusion}.

\begin{figure}[h]

\tikzset{every picture/.style={line width=0.75pt}} %set default line width to 0.75pt        
\centering
\begin{tikzpicture}[x=0.75pt,y=0.75pt,yscale=-1,scale=0.8]
%uncomment if require: \path (0,559); %set diagram left start at 0, and has height of 559

%Shape: Ellipse [id:dp24188024432234023] 
\draw   (319.73,191.87) .. controls (396.85,191.82) and (459.41,240.25) .. (459.45,300.05) .. controls (459.5,359.86) and (397.01,408.38) .. (319.88,408.43) .. controls (242.75,408.49) and (180.19,360.05) .. (180.15,300.25) .. controls (180.11,240.45) and (242.6,191.92) .. (319.73,191.87) -- cycle ;
%Shape: Ellipse [id:dp7760255396062348] 
\draw   (400.77,284.41) .. controls (477.9,284.36) and (540.46,332.79) .. (540.5,392.6) .. controls (540.54,452.4) and (478.05,500.92) .. (400.92,500.97) .. controls (323.8,501.03) and (261.24,452.59) .. (261.2,392.79) .. controls (261.15,332.99) and (323.64,284.47) .. (400.77,284.41) -- cycle ;
%Shape: Ellipse [id:dp5125617013392456] 
\draw   (238.83,283.06) .. controls (315.96,283) and (378.52,331.44) .. (378.56,391.24) .. controls (378.6,451.05) and (316.11,499.57) .. (238.98,499.62) .. controls (161.86,499.68) and (99.3,451.24) .. (99.26,391.44) .. controls (99.21,331.64) and (161.7,283.11) .. (238.83,283.06) -- cycle ;

% Text Node
\draw  [color={rgb, 255:red, 255; green, 255; blue, 255 }  ,draw opacity=1 ][fill={rgb, 255:red, 255; green, 255; blue, 255 }  ,fill opacity=1 ]  (288.23,180.87) -- (350.23,180.87) -- (350.23,202.87) -- (288.23,202.87) -- cycle  ;
\draw (319.73,191.87) node  [align=left] {$i$-MRD};
% Text Node
\draw  [color={rgb, 255:red, 255; green, 255; blue, 255 }  ,draw opacity=1 ][fill={rgb, 255:red, 255; green, 255; blue, 255 }  ,fill opacity=1 ]  (64.26,380.44) -- (134.26,380.44) -- (134.26,402.44) -- (64.26,402.44) -- cycle  ;
\draw (99.26,391.44) node  [align=left] {$(i+1)$-BMD};
% Text Node
\draw  [color={rgb, 255:red, 255; green, 255; blue, 255 }  ,draw opacity=1 ][fill={rgb, 255:red, 255; green, 255; blue, 255 }  ,fill opacity=1 ]  (503.5,381.6) -- (577.5,381.6) -- (577.5,403.6) -- (503.5,403.6) -- cycle  ;
\draw (540.5,392.6) node  [align=left] {$(i-1)$-MRD};
% Text Node
%\draw  [dash pattern={on 4.5pt off 4.5pt}]  (321, 250) circle [x radius= 12.5, y radius= 12.5]   ;
\draw (321,250) node  [align=left] {(1)};
% Text Node
%\draw  [dash pattern={on 4.5pt off 4.5pt}]  (241, 307) circle [x radius= 12.5, y radius= 12.5]   ;
\draw (241,307) node  [align=left] {(2)};
% Text Node
%\draw  [dash pattern={on 4.5pt off 4.5pt}]  (320, 337) circle [x radius= 12.5, y radius= 12.5]   ;
\draw (320,337) node  [align=left] {(3)};
% Text Node
%\draw  [dash pattern={on 4.5pt off 4.5pt}]  (405, 307) circle [x radius= 12.5, y radius= 12.5]   ;
\draw (405,307) node  [align=left] {(4)};
% Text Node
%\draw  [dash pattern={on 4.5pt off 4.5pt}]  (195, 420) circle [x radius= 12.5, y radius= 12.5]   ;
\draw (195,420) node  [align=left] {(5)};
% Text Node
%\draw  [dash pattern={on 4.5pt off 4.5pt}]  (446, 420) circle [x radius= 12.5, y radius= 12.5]   ;
\draw (446,420) node  [align=left] {(7)};
% Text Node
%\draw  [dash pattern={on 4.5pt off 4.5pt}]  (321, 440) circle [x radius= 12.5, y radius= 12.5]   ;
\draw (321,440) node  [align=left] {(6)};
% Text Node
\draw (320.39,388.05) node  [align=left] {$i$-BMD};
% Text Node
\draw (321.39,374.05) node  [align=left] {\small not minimally};
% Text Node
\draw (237.39,335.05) node  [align=left] {\small minimally};
% Text Node
\draw (236.39,349.05) node  [align=left] {$i$-BMD};

\end{tikzpicture}
\caption{ Relations between $i$-MRD and $i$-BMD codes.}\label{FigInclusion}
\end{figure}
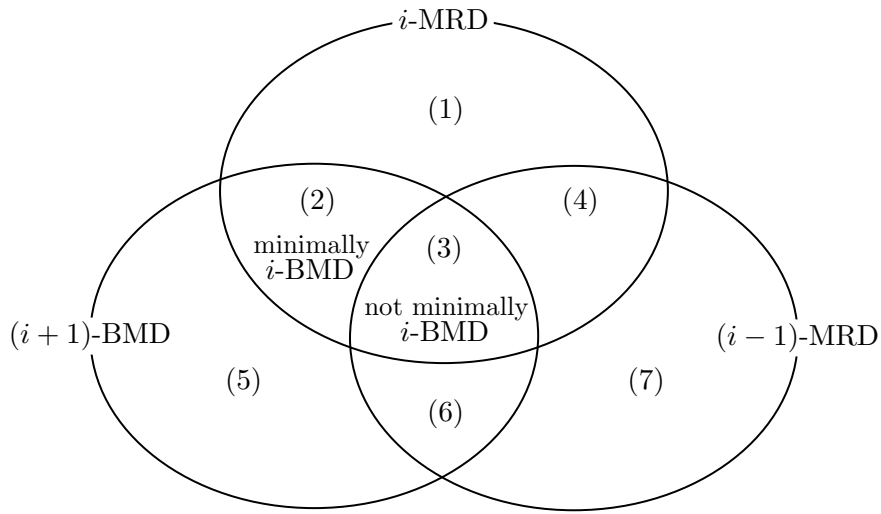

Observe that, in general, all the regions in Figure~\ref{FigInclusion} are non-empty. We illustrate this for $i=2$ giving an example of codes for each region.

\begin{enumerate}
	\item The code $\mathcal{C}_1$ in Example \ref{Example(i)not(i+1)MRD} is $2$-MRD, not $1$-MRD and not $3$-BMD.

\medskip

	\item The code $\mathcal{C}_2$ of dimension $6$ generated by
			\begin{equation*}
				\footnotesize
				\left\{\begin{pmatrix}
                1 & 0 & 0 & 0\\
                1 & 0 & 1 & 0\\
                0 & 0 & 0 & 0
                \end{pmatrix},
                \begin{pmatrix}
                0 & 1 & 0 & 0\\
                1 & 0 & 0 & 1\\
                1 & 1 & 0 & 1
                \end{pmatrix},
                \begin{pmatrix}
                0 & 0 & 1 & 0\\
                1 & 0 & 0 & 1\\
                0 & 1 & 0 & 1
                \end{pmatrix},
                \begin{pmatrix}
                0 & 0 & 0 & 1\\
                0 & 0 & 1 & 0\\
                0 & 1 & 0 & 0
                \end{pmatrix},
                \begin{pmatrix}
                0 & 0 & 0 & 0\\
                0 & 1 & 0 & 1\\
                0 & 1 & 0 & 0
                \end{pmatrix},
                \begin{pmatrix}
                0 & 0 & 0 & 0\\
                0 & 0 & 0 & 0\\
                0 & 0 & 1 & 0
                \end{pmatrix}\right\}
			\end{equation*} 
			 is minimally $2$-BMD, $2$-MRD, $3$-BMD but not $1$-MRD.

\medskip

	\item The code $\mathcal{C}_3$ of dimension $6$ generated by
			\begin{equation*}
			\footnotesize
				\left\{\begin{pmatrix}
                1 & 0 & 0 & 0\\
                0 & 1 & 1 & 0\\
                0 & 1 & 0 & 0
                \end{pmatrix},
                \begin{pmatrix}
                0 & 1 & 0 & 0\\
                0 & 1 & 0 & 0\\
                1 & 0 & 1 & 1
                \end{pmatrix},
                \begin{pmatrix}
                0 & 0 & 1 & 0\\
                0 & 1 & 0 & 0\\
                0 & 1 & 0 & 1
                \end{pmatrix},
                \begin{pmatrix}
                0 & 0 & 0 & 1\\
                0 & 0 & 1 & 0\\
                1 & 1 & 0 & 1
                \end{pmatrix},
                \begin{pmatrix}
                0 & 0 & 0 & 0\\
                1 & 1 & 1 & 0\\
                0 & 1 & 1 & 0
                \end{pmatrix},
                \begin{pmatrix}
                0 & 0 & 0 & 0\\
                0 & 0 & 0 & 1\\
                1 & 1 & 0 & 0
                \end{pmatrix}\right\}
			\end{equation*}
			is $1$-BMD and therefore $2$-BMD, $2$-MRD and $3$-BMD.

\medskip

	\item The code $\mathcal{C}_4$ of dimension $6$ generated by	
			\begin{equation*}
				\footnotesize
				\left\{
				\begin{pmatrix}
                1 & 0 & 0 & 0\\
                0 & 0 & 0 & 1\\
                0 & 0 & 1 & 0
                \end{pmatrix},
                \begin{pmatrix}
                0 & 1 & 0 & 0\\
                0 & 0 & 1 & 0\\
                0 & 0 & 1 & 1
                \end{pmatrix},
                \begin{pmatrix}
                0 & 0 & 1 & 0\\
                0 & 0 & 0 & 1\\
                0 & 1 & 1 & 0
                \end{pmatrix},
                \begin{pmatrix}
                0 & 0 & 0 & 1\\
                0 & 0 & 1 & 1\\
                1 & 0 & 0 & 0
                \end{pmatrix},
                \begin{pmatrix}
                0 & 0 & 0 & 0\\
                1 & 0 & 0 & 1\\
                0 & 1 & 0 & 0
                \end{pmatrix},
                \begin{pmatrix}
                0 & 0 & 0 & 0\\
                0 & 1 & 0 & 0\\
                1 & 0 & 1 & 0
                \end{pmatrix}\right\}
			\end{equation*}
			is $1$-MRD, $2$-MRD but not $3$-BMD.

\medskip

	\item The code $\mathcal{C}_5$ of dimension $6$ generated by
			\begin{equation*}
			\footnotesize
				\left\{
				\begin{pmatrix}
                1 & 0 & 0 & 0\\
                0 & 0 & 0 & 1\\
                0 & 1 & 1 & 1
                \end{pmatrix},
                \begin{pmatrix}
                0 & 1 & 0 & 0\\
                0 & 0 & 0 & 0\\
                1 & 1 & 0 & 0
                \end{pmatrix},
                \begin{pmatrix}
                0 & 0 & 1 & 0\\
                0 & 0 & 0 & 1\\
                0 & 0 & 1 & 1
                \end{pmatrix},
                \begin{pmatrix}
                0 & 0 & 0 & 1\\
                0 & 0 & 1 & 0\\
                0 & 0 & 0 & 0
                \end{pmatrix},
                \begin{pmatrix}
                0 & 0 & 0 & 0\\
                1 & 0 & 1 & 1\\
                0 & 1 & 0 & 0
                \end{pmatrix},
                \begin{pmatrix}
                0 & 0 & 0 & 0\\
                0 & 1 & 0 & 1\\
                0 & 0 & 1 & 0
                \end{pmatrix}\right\}
			\end{equation*}
			is $3$-BMD, not $1$-MRD and not $2$-MRD.

\medskip

	\item The code $\mathcal{C}_6$ of dimension $5$ generated by
			\begin{equation*}
			\footnotesize
				\left\{
				\begin{pmatrix}
                1 & 0 & 1 & 0\\
                0 & 0 & 1 & 0\\
                0 & 0 & 0 & 1
                \end{pmatrix},
                \begin{pmatrix}
                0 & 1 & 1 & 0\\
                0 & 0 & 0 & 0\\
                1 & 1 & 0 & 0
                \end{pmatrix},
                \begin{pmatrix}
                0 & 0 & 0 & 1\\
                0 & 0 & 1 & 1\\
                1 & 0 & 0 & 0
                \end{pmatrix},
                \begin{pmatrix}
                0 & 0 & 0 & 0\\
                1 & 0 & 1 & 1\\
                1 & 0 & 0 & 0
                \end{pmatrix},
                \begin{pmatrix}
                0 & 0 & 0 & 0\\
                0 & 1 & 1 & 1\\
                0 & 0 & 1 & 0
                \end{pmatrix}\right\}
			\end{equation*}
			is $1$-MRD, $3$-BMD but not $2$-MRD.

\medskip

	\item The code $\mathcal{C}_7$ of dimension $5$ generated by
			\begin{equation*}
				\footnotesize
				\left\{
				\begin{pmatrix}
                1 & 0 & 0 & 0\\
                0 & 1 & 1 & 0\\
                0 & 0 & 0 & 1
                \end{pmatrix},
                \begin{pmatrix}
                0 & 1 & 0 & 0\\
                0 & 0 & 1 & 0\\
                0 & 1 & 0 & 0
                \end{pmatrix},
                \begin{pmatrix}
                0 & 0 & 1 & 0\\
                0 & 1 & 0 & 0\\
                1 & 0 & 0 & 1
                \end{pmatrix},
                \begin{pmatrix}
                0 & 0 & 0 & 1\\
                0 & 0 & 1 & 0\\
                1 & 1 & 0 & 0
                \end{pmatrix},
                \begin{pmatrix}
                0 & 0 & 0 & 0\\
                0 & 0 & 0 & 1\\
                1 & 1 & 1 & 0
                \end{pmatrix}\right\}
			\end{equation*} 
			is $1$-MRD, not $2$-MRD and not $3$-BMD.
\end{enumerate}

We conclude this section with an example of two minimally $i$-BMD codes of the same dimension with the property that their dual codes are not minimally $j$-BMD for the same $0\leq j\leq k^\perp$. This shows that the property of being minimally $i$-BMD does not obey a duality statement.

\begin{example}
\label{Examplenotduality}
	Let $\mathcal{C}_1$ be the code of Example \ref{Example(i)not(i+1)MRD} and $\mathcal{C}_4$ the code defined above. One can check that the generalized rank weight distribution of $\mathcal{C}_4$ is
    \begin{equation*}
    	d_1(\mathcal{C}_4)=d_2(\mathcal{C}_4)=d_3(\mathcal{C}_4)=2, \qquad d_4(\mathcal{C}_4)=d_5(\mathcal{C}_4)=d_6(\mathcal{C}_4)=3,
    \end{equation*} 
    the generalized rank weight distribution of $\mathcal{C}_1^\perp$ is
    \begin{equation*}
    	d_1(\mathcal{C}_1^\perp)=1, \qquad d_2(\mathcal{C}_1^\perp)=d_3(\mathcal{C}_1^\perp)=2, \qquad d_4(\mathcal{C}_1^\perp)=d_5(\mathcal{C}_1^\perp)=d_6(\mathcal{C}_1^\perp)=3,
    \end{equation*}
    the generalized rank weight distribution of $\mathcal{C}_4^\perp$ is
    \begin{equation*}
    	d_1(\mathcal{C}_4^\perp)=1, \qquad d_2(\mathcal{C}_4^\perp)=2,\qquad d_3(\mathcal{C}_4^\perp)=d_4(\mathcal{C}_4^\perp)=d_5(\mathcal{C}_4^\perp)=d_6(\mathcal{C}_4^\perp)=3
    \end{equation*}
    and
    \begin{equation*}
    	\begin{aligned}
    		n-d_1(\mathcal{C}_1^\perp)-d_3(\mathcal{C}_1)=n-d_1(\mathcal{C}_4^\perp)-d_3(\mathcal{C}_4)&=3-2-1=0\geq 0,\\
    		n-d_1(\mathcal{C}_1^\perp)-d_4(\mathcal{C}_1)=n-d_1(\mathcal{C}_4^\perp)-d_4(\mathcal{C}_4)&=3-3-1=-1<0,\\
    		n-d_1(\mathcal{C}_1)-d_3(\mathcal{C}_1^\perp)&=3-2-1=0\geq 0,\\
    		n-d_1(\mathcal{C}_1)-d_4(\mathcal{C}_1^\perp)&=3-3-1=-1<0,\\
    		n-d_1(\mathcal{C}_2)-d_2(\mathcal{C}_4^\perp)&=3-2-1=0\geq0,\\
    		n-d_1(\mathcal{C}_2)-d_3(\mathcal{C}_4^\perp)&=3-3-1=-1<0.\\
    	\end{aligned}
    \end{equation*}
    Therefore, we get that both $\mathcal{C}_1$ and $\mathcal{C}_4$ are minimally $4$-BMD, $\mathcal{C}_1^\perp$ is minimally $4$-BMD, while~$\mathcal{C}_4^\perp$ is minimally $2$-BMD.
\end{example}

\section{The Generalized Zeta Function}
\label{SectionZeta}
Inspired by the work in~\cite{blanco2018rank}, in this section we define and study the zeta function for generalized rank weights. Throughout this section we work in the polynomial rings~$\mathbb{Q}[T]$ and $\mathbb{Q}[X,Y,T]$. 

\begin{definition}
	The \textbf{$i$-th generalized zeta function} $Z_\mathcal{C}^{(i)}(T)$ of $\mathcal{C}$ is the generating function of the $i$-th generalized normalized binomial moments of $\mathcal{C}$, i.e.,
	\begin{equation}
	\label{EqZeta}
		Z_\mathcal{C}^{(i)}(T):=\sum_{u\in\mathbb{Z}}b_u^{(i)}(\mathcal{C})T^u=\sum_{u\geq 0}b_u^{(i)}(\mathcal{C})T^u.
	\end{equation}
\end{definition}

The following result is the rank-metric analogue of \cite[Theorem~3.8]{jurrius2012codes}

\begin{theorem}
\label{TheoremZetaPol}
	There exists a unique polynomial $P_\mathcal{C}^{(i)}(T) \in \mathbb{Q}[T]$ 
	such that 
	\begin{equation}
	\label{EqZetaP}
		Z_\mathcal{C}^{(i)}(T)=\frac{P_\mathcal{C}^{(i)}(T)}{\prod_{j=0}^i(1-q^{mj}T)}.
	\end{equation}
	More precisely, the coefficient of $T^u$ in $P_\mathcal{C}^{(i)}(T)$ is 
	\begin{equation}
	\label{EqCoeffP}
		p_u^{(i)}(\mathcal{C})=\sum_{j=0}^{i+1}\qmbin{i+1}{j}(-1)^jq^{m\binom{j}{2}}b_{u-j}^{(i)}(\mathcal{C})
	\end{equation}
	and the degree of $P_\mathcal{C}^{(i)}(T)$ is at most $n-d^\perp-d_i+i+1$.
\end{theorem}
\begin{proof}
	Define the power series
	\begin{equation*}
		P_\mathcal{C}^{(i)}(T):=Z_\mathcal{C}(T)\prod_{j=0}^i(1-q^{mj}T).
	\end{equation*} 
	Using Lemma~\ref{LemmaTools} one can show that the coefficient of $T^u$ in $P_\mathcal{C}^{(i)}(T)$ is exactly the quantity in~\eqref{EqCoeffP}. Let $u>n-d^\perp-d_i+i+1$. We have
	\begin{equation*}
		b_{u-j}^{(i)}(\mathcal{C})=\qbin{k-m(n-u+j-d_j)}{j}
	\end{equation*}
	by definition. A standard computation using Lemma~\ref{LemmaTools} shows that
	\begin{equation*}
	\sum_{j=0}^{i+1}\qmbin{i+1}{j}(-1)^jq^{m\binom{j}{2}}b_{u-j}^{(i)}(\mathcal{C})=0.
	\end{equation*}
%	We remand to Lemma \ref{LemmaPu} for the complete proof of this last result. 
	Therefore $P_\mathcal{C}^{(i)}(T)$ is a polynomial and its degree is at most $n-d^\perp-d_i+i+1$.
\end{proof}

We call the polynomial $P_\mathcal{C}^{(i)}(T)$ in~\eqref{EqZetaP} the \textbf{$i$-th generalized zeta polynomial} of~$\mathcal{C}$.

It is interesting to observe that for an $i$-BMD code the generalized zeta polynomials and zeta functions are partially determined by the code dimension, as the next result shows. The proof easily follows from  Theorem~\ref{TheoremiBMD=>iMRD}.

\begin{corollary} 
\label{cor:several}
	Suppose that $\mathcal{C}$ is $i$-BMD. The following hold for $i\leq j\leq k$.
	\begin{enumerate}
		\item $b_u^{(j)}(\mathcal{C})=\qbin{\rho+m\left(u-\left\lfloor\frac{\rho-j}{m}\right\rfloor\right)}{j}$, for $u\geq 0$.
		\item $\displaystyle Z_\mathcal{C}^{(j)}(T)=\sum_{u\geq 0}\qbin{\rho+m\left(u-\left\lfloor\frac{\rho-j}{m}\right\rfloor\right)}{j}T^u$.
		\item $\displaystyle p_u^{(j)}(\mathcal{C})=\sum_{t=0}^{j+1}\qmbin{j+1}{t}(-1)^tq^{m\binom{t}{2}}\qbin{\rho+m\left(u-t-\left\lfloor\frac{\rho-j}{m}\right\rfloor\right)}{j}$, for $u\geq 0$.
		\item $\deg(P_\mathcal{C}^{(j)}(T))\leq \min\left\{j,\left\lfloor\frac{k-j}{m}\right\rfloor+j+1-d^\perp\right\}$.
	\end{enumerate}
\end{corollary}

Note that $b_u^{(j)}(\mathcal{C})$ for $u\geq 0$ only depends on $\rho$. This implies that if $\mathcal{C}$ and $\mathcal{D}$ are $i$-BMD codes such that $\dimq{\mathcal{C}}\equiv\dimq{\mathcal{D}}\equiv\rho \mod m$, then $b_u^{(j)}(\mathcal{C})=b_u^{(j)}(\mathcal{D})$ for every $u\geq 0$. 

\begin{remark}
\label{RemarkbuBMD}
	Corollary~\ref{cor:several} implies that if $\mathcal{C}$ is $1$-BMD and $m|k$, then $P_\mathcal{C}^{(1)}(T)$ is a polynomial of degree $0$. More precisely,
	\begin{equation*}
		P_\mathcal{C}^{(1)}(T)=\qbin{m}{1}.
	\end{equation*}
	This recovers \cite[Lemma~4]{blanco2018rank}.
\end{remark}

\begin{notation} \label{not:things}
	In the sequel, for $j,u\in\mathbb{Z}_{\geq 0}$ and $\tau\in\{0,\ldots,m-1\}$ we let
	\begin{equation*}
	\begin{aligned}
		b_{\tau,u}^{(j)}&:=\qbin{\tau+m\left(u-\left\lfloor\frac{\tau-j}{m}\right\rfloor\right)}{j},\\
		p_{\tau,u}^{(j)}&:=\sum_{t=0}^{j+1}\qmbin{j+1}{t}(-1)^tq^{m\binom{t}{2}}\qbin{\tau+m\left(u-t-\left\lfloor\frac{\tau-j}{m}\right\rfloor\right)}{j},\\
		Z_{\tau}^{(j)}(T)&:=\sum_{u\geq 0}b_{\tau,u}^{(j)}T^u,\\
		P_{\tau}^{(i)}(T)&:=\sum_{u\geq 0}p_{\tau,u}^{(j)}T^u.
	\end{aligned}
	\end{equation*}
\end{notation}

Observe that, by Corollary~\ref{cor:several}, the objects $b_{\tau,u}^{(j)}$,
 $p_{\tau,u}^{(j)}$,
 $Z_{\tau}^{(j)}(T)$,
$P_{\tau}^{(i)}(T)$
in Notation~\ref{not:things} are those associated with 
a $j$-BMD code of dimension $k \equiv \tau \mod m$,
provided that such a code exists.

The next result shows a precise connection, via the $q$-Bernstein polynomials, among the $i$-th generalized zeta function, the $i$-th generalized rank weight enumerator, and the $i$-th generalized rank weight.

\begin{lemma}
\label{LemmaWivarphi}
	Let $\varphi_n(X,Y,T)$ be the function defined by
	\begin{equation*}
		\varphi_n(X,Y,T):=\sum_{u=0}^n\mathcal{B}_{n,u}(X,Y;q)T^{n-u}.
	\end{equation*}
	Then $W_\mathcal{C}^{(i)}(X,Y)$ is the coefficient of $T^{n-d_i}$ in the expression $Z_\mathcal{C}^{(i)}(T)\varphi_n(X,Y,T)$.
\end{lemma}
\begin{proof}
	We have
	\begin{align*}
			Z_\mathcal{C}^{(i)}(T)\varphi_n(X,Y,T)&=\left(\sum_{u\geq 0}b_u^{(i)}(\mathcal{C})T^u\right)\left(\sum_{u=0}^n\mathcal{B}_{n,u}(X,Y;q)T^{n-u}\right)\\
			&=\sum_{u\geq 0}\sum_{t=0}^nb_u^{(i)}(\mathcal{C})\mathcal{B}_{n,t}(X,Y;q)T^{n-t+u}\\
			&\equiv \sum_{t=0}^nT^t\sum_{u=0}^tb_u^{(i)}(\mathcal{C})\mathcal{B}_{n,n-t+u}(X,Y;q) \mod T^{n+1}\\
			&\equiv \sum_{t=0}^nT^{n-t}\sum_{u=0}^{n-t}b_u^{(i)}(\mathcal{C})\mathcal{B}_{n,t+u}(X,Y;q) \mod T^{n+1}.
		\end{align*}
	Therefore the coefficient of $T^{n-d_i}$ in $Z_\mathcal{C}^{(i)}(T)\varphi_n(X,Y,T)$ is
	\begin{equation*}
		\sum_{u=0}^{n-d_i}b_u^{(i)}(\mathcal{C})\mathcal{B}_{n,u+d_i}(X,Y;q)=W_\mathcal{C}^{(i)}(X,Y),
	\end{equation*}
	where the last equality follows from Theorem~\ref{TheoremWBern}.
\end{proof}

\begin{example}
	Let $\mathcal{C}_1$ be the code of Example \ref{Example(i)not(i+1)MRD} and recall that $d_2(\mathcal{C}_1)=2$. One can check that
	\begin{equation*}
		\begin{aligned}
		\varphi_n(X,Y,T) \ = \ &Y^3+(7XY^2-7Y^3)T+(7X^2Y-21XY^2+14Y^3)T^2\\
				&+(X^3-7X^2Y+14XY^2-8Y^3)T^3,\\
		Z_{\mathcal{C}_1}^{(2)}(T)\ =\ &\frac{13}{7}+651T+174251T^2+44731051T^3+11453115051T^4+\ldots, \\
		W_{\mathcal{C}_1}^{(2)}(X,Y) \ = \ &13XY^2 + 638Y^3.
		\end{aligned}
	\end{equation*}
	Therefore $W_{\mathcal{C}_1}^{(2)}(X,Y)$ is the coefficient of $T^{n-d_2(\mathcal{C}_1)}=T$ in
	\begin{equation*}
		\begin{aligned}
			Z_{\mathcal{C}_1}^{(2)}(T)\varphi_n(X,Y,Z)=&\frac{13}{7}Y^3+(13XY^2 + 638Y^3)T\\& +(13X^2Y + 4518XY^2 + 169720Y^3)T^2+\ldots,
		\end{aligned}
	\end{equation*}
	as predicted by Lemma~\ref{LemmaWivarphi}.
\end{example}

We conclude this section with a result on generalized zeta functions/polynomials that we will need later.

\begin{lemma}
\label{LemmaZMRDP}
	For all $0\leq \tau\leq m-1$ we have
	\begin{equation*}
		Z_{\tau}^{(i)}(T)P_\mathcal{C}^{(i)}(T)= Z_\mathcal{C}^{(i)}(T)P_{\tau}^{(i)}(T).
	\end{equation*}
\end{lemma}
\begin{proof}
	The result follows from Theorem \ref{TheoremZetaPol}. Indeed,
	\begin{equation*}
		Z_{\tau}^{(i)}(T)P_\mathcal{C}^{(i)}(T)=\frac{P_\tau^{(i)}(T)}{\prod_{j=0}^i(1-q^{mj}T)}P_\mathcal{C}^{(i)}(T)=\frac{P_\mathcal{C}^{(i)}(T)}{\prod_{j=0}^i(1-q^{mj}T)}P_\tau^{(i)}(T)=Z_\mathcal{C}^{(i)}(T)P_{\tau}^{(i)}(T),
	\end{equation*}
	as desired.
\end{proof}

\section{A Connection with Bell Polynomials}
\label{SectionBell}

For each $\tau\in\{0,\ldots,m-1\}$ and $r\in\{0,\ldots,n\}$ we define
\begin{equation*}
		M_{\tau,r}^{(i)}(X,Y):=\sum_{u=0}^{n-r}\qbin{\tau+m\left(u-\left\lfloor\frac{\tau-i}{m}\right\rfloor\right)}{i}\mathcal{B}_{n,u+r}(X,Y;q).
	\end{equation*}
Note that, by Corollary \ref{cor:several}, $M_{\tau,r}^{(i)}(X,Y)$ is the $i$-th generalized rank weight enumerator of an $i$-BMD  code of dimension $k\equiv\tau\mod m$ and~$d_i=r$, provided that such a code exists.
	
It is not difficult to check that, for a given $\tau$, the set $\mathcal{M}_\tau^{(i)}:=\{M_{\tau,r}^{(i)}(X,Y): 0\leq r\leq n\}$ is a~$\mathbb{Q}$-basis for the space that contains the $i$-th generalized rank weight enumerators.

In \cite[Theorem~1]{blanco2018rank}, the authors show for $i=1$ that it is possible to express the $1$st generalized rank weight enumerator of $\mathcal{C}$ with respect to the basis $\mathcal{M}_0^{(1)}$ as follows
\begin{equation*}
	W_\mathcal{C}^{(1)}(X,Y)=\sum_{j=0}^{n-d}p_{j}^{(1)}(\mathcal{C}) \, M_{0,d+j}^{(1)}(X,Y),
\end{equation*}
where the $p_{j}^{(1)}(\mathcal{C})$ are the coefficients of the $1$st generalized zeta polynomial $P_{\mathcal{C}}^{(1)}(T)$.

The aim of this section is to extend this result by computing the coefficients of an $i$-th generalized rank weight enumerator with respect to the basis $\mathcal{M}_\tau^{(i)}$. In particular, we give an explicit expression for the rational numbers $\beta_{\tau,j}^{(i)}(\mathcal{C})$ such that 
\begin{equation*}
	W_\mathcal{C}^{(i)}(X,Y)=\sum_{j=0}^{n-d_i}\beta_{\tau,j}^{(i)}(\mathcal{C}) \, M_{\tau,d_i+j}^{(i)}(X,Y).
\end{equation*}
and show the relation with $P_\mathcal{C}^{(i)}(T)$.

Our approach for deriving such coefficients is based on Bell polynomials. Introduced in~\cite{bell1927partition}, these polynomials encode the different ways in which an integer can be partitioned. See also~\cite{comtet1974advanced} and~\cite{port1994polynomial}. These polynomials find applications in several fields of mathematics.

There are two forms of Bell polynomials, namely the exponential and the ordinary form. For convenience, we introduce an homogeneous version of the definition of the ordinary Bell polynomials. In the following, $a,b$ are non-negative integers.

\begin{definition}
	The \textbf{homogeneous ordinary partial Bell polynomials} are the polynomials $\mathcal{P}_{u,w}(x_0,x_1,\ldots,x_{a-b+1})$  in an infinite number of variables $x_0,x_1,\ldots$, defined by the formula 
	\begin{equation*}
		\Phi(X,Y)=\exp\left(Y\sum_{j\geq 1} x_jx_0^{j-1}X^j \right)=\sum_{0\leq w\leq u}\mathcal{P}_{u,w}(x_0,x_1,\ldots,x_{u-w+1})X^u\frac{Y^w}{w!}
	\end{equation*}
	or, equivalently, by the series expansion as 
	\begin{equation*}
		\left(\sum_{j\geq 1} x_jx_0^{j-1}X^j\right)^w=\sum_{u\geq w}\mathcal{P}_{u,w}(x_0,x_1,\ldots,x_{u-w+1})X^u.
	\end{equation*} 
\end{definition}

An explicit way to compute such polynomials is given by the following well-known result. 

\begin{lemma}
	The \textbf{$(a,b)$-th homogeneous ordinary partial Bell polynomial} is 
	\begin{align*}
		\mathcal{P}_{a,b}(x_0,x_1,\ldots,x_{a-b+1})&=\sum\frac{k!}{j_1!j_2!\cdots j_{a-b+1}!}x_1^{j_1}\left(x_2x_0\right)^{j_2}\cdots\left(x_{a-b+1}x_0^{a-b}\right)^{j_{a-b+1}}\\
		&=\sum\frac{k!}{j_1!j_2!\cdots j_{a-b+1}!}x_0^{a-b}x_1^{j_1}x_2^{j_2}\cdots x_{a-b+1}^{j_{a-b+1}}
	\end{align*}
	where the sum ranges over
	\begin{align*}
			j_1+j_2+\cdots+j_{a-b+1}&=b,\\
			j_1+2j_2+\cdots+(a-b+1)j_{a-b+1}&=a.
		\end{align*}
	Moreover, $\mathcal{P}_{0,0}:=1$, $\mathcal{P}_{a,0}:=0$ and $\mathcal{P}_{0,b}:=0$.
\end{lemma}

(Compare with the definition of exponential partial Bell polynomial and Theorem A given in~\cite[Section~3.3]{comtet1974advanced}).
The $(a,b)$-th homogeneous ordinary partial Bell polynomial describes the compositions of the integer $a$ in $b$ parts\footnote{The compositions of an integer $a$ are the possible ways the integer $a$ can be partitioned in where the order matters. In other words, compositions are simply ordered partitions.}.

\begin{definition}
	The \textbf{$a$-th homogeneous ordinary Bell polynomial} $\mathcal{P}_{a}(x_0,x_1,\ldots,x_{a})$  is the sum, over $b$, of all the $(a,b)$-th homogeneous ordinary partial Bell polynomials, i.e.
	\begin{equation*}
		\mathcal{P}_{a}(x_0,x_1,\ldots,x_{a}):=\sum_{b=0}^a \mathcal{P}_{a,b}(x_0,x_1,\ldots,x_{a-b+1}).
	\end{equation*}
\end{definition}

\begin{example}
	In this example we compute the $5$-th homogeneous ordinary Bell polynomial by summing all the $(5,b)$-th homogeneous ordinary partial Bell polynomials, $1\leq b\leq 5$. Consider $b=3$, for example; we have the following compositions
	\begin{equation*}
		1+1+3\qquad 1+3+1 \qquad 3+1+1 \qquad 1+2+2 \qquad 2+1+2 \qquad 2+2+1.
	\end{equation*}
	Notice that the first three compositions give the monomial $3x_1^2x_3x_0^2$ while the last three give~$3x_1x_2^2x_0^2$. Therefore, 
	\begin{equation*}
		\mathcal{P}_{5,3}(x_0,x_1,x_2,x_3)=(3x_1^2x_3+3x_1x_2^2)x_0^2
	\end{equation*}
	(recall that the degree of $x_0$ is $5-b$).
	We can compute the remaining $(5,b)$-th homogeneous ordinary partial Bell polynomials analogously and by summing them up over $b$ we get that the $5$-th homogeneous ordinary Bell polynomial is 
	\begin{equation*}
		\mathcal{P}_{5}(x_0,x_1,x_2,x_3,x_4,x_5)=x_1^5+4x_1^3x_2x_0+(3x_1^2x_3+3x_1x_2^2)x_0^2+(2x_1x_4+2x_2x_3)x_0^3+x_5x_0^4.
	\end{equation*}

\end{example}

The following result is an application of Fa\`a di Bruno's formula~\cite{faa1855sullo}. In this paper
we need a purely combinatorial formualtion such result, involving the Bell polynomial.
It can be found in \cite{riordan1946derivatives}.

\begin{lemma}[{\cite[Section~3.5]{comtet1974advanced}}]
	\label{LemmaBell}
	If $(c_u)_{u\in\mathbb{Z}_{\geq 0}}$ is a sequence and $Q(T)$ is the formal power series
	\begin{equation*}
		Q(T)=1-\sum_{u\geq 1} \frac{c_u}{c_0}T^u
	\end{equation*}
then its reciprocal $\displaystyle Y(T) = \frac{1}{Q(T)}$ can be written as $\displaystyle Y(T) = 1 + \sum_{a\geq 1} y_aT^a$ with	
	\begin{equation*}
		y_a =\left(\frac{1}{c_0}\right)^a\sum_{b=0}^a \mathcal{P}_{a,b}(c_0,c_1, c_2,\ldots,c_{a-b+1})=\left(\frac{1}{c_0}\right)^a\mathcal{P}_a(c_0,c_1, c_2, \ldots,c_n).
	\end{equation*}
\end{lemma}
Note that the last equality in the previous lemma follows from the definition
of homogeneous ordinary Bell polynomials and the fact that $\mathcal{P}_{a,0}=0$ for all $a\geq 0$.
%}

We now state the main theorem of this section. In Corollary \ref{CorollaryW=betaM}, this will provide us with a way of expanding~$W_{\mathcal{C}}^{(i)}(X,Y)$ over the basis $\mathcal{M}_\tau^{(i)}$.

\begin{theorem}
\label{TheoremZBeta}
	For all $0\leq \tau\leq m-1$ we have
	\begin{equation}
	\label{EqZBeta}
		Z_\mathcal{C}^{(i)}(T)= Z_{\tau}^{(i)}(T)\sum_{u\geq 0}\beta_{\tau,u}^{(i)}(\mathcal{C})T^u,
	\end{equation}
	where
	\begin{equation*}
		\beta_{\tau,u}^{(i)}(\mathcal{C}):=\sum_{j=0}^u\frac{p_j^{(i)}(\mathcal{C})}{\left(p_{\tau,0}^{(i)}\right)^{u-j+1}}\mathcal{P}_{u-j}\left(p_{\tau,0}^{(i)},-p_{\tau,1}^{(i)},-p_{\tau,2}^{(i)},\ldots,-p_{\tau,u-j}^{(i)}\right).
	\end{equation*}
\end{theorem}
\begin{proof}
	We start by observing that
	\begin{equation*}
		\frac{P_\tau^{(i)}(T)}{p_{\tau,0}^{(i)}}=\sum_{u=0}^{n-d^\perp-d_i+i+1}\frac{p_{\tau,u}^{(i)}}{p_{\tau,0}^{(i)}}T^u=1-\sum_{u\geq 1}\left(-\frac{p_{\tau,u}^{(i)}}{p_{\tau,0}^{(i)}}\right)T^u,
	\end{equation*}
	as $p_{\tau,u}^{(i)}=0$ for all $u>n-d^\perp-d_i+i+1$. Applying Lemma \ref{LemmaBell} we obtain 
	\begin{align*}
		\left(\frac{P_\tau^{(i)}(T)}{p_{\tau,0}^{(i)}}\right)^{-1}&=1+\sum_{u\geq 1} \left(\frac{1}{p_{\tau,0}^{(i)}}\right)^u\mathcal{P}_u\left(p_{\tau,0}^{(i)},-p_{\tau,1}^{(i)},-p_{\tau,2}^{(i)},\ldots,-p_{\tau,u}^{(i)}\right)T^u\\
		&=\sum_{u\geq 0} \left(\frac{1}{p_{\tau,0}^{(i)}}\right)^u\mathcal{P}_u\left(p_{\tau,0}^{(i)},-p_{\tau,1}^{(i)},-p_{\tau,2}^{(i)},\ldots,-p_{\tau,u}^{(i)}\right)T^u,
		\end{align*}
	as $\mathcal{P}_0=\mathcal{P}_{0,0}=1$ by definition. Combining this with Lemma \ref{LemmaZMRDP} we obtain
	\begin{align*}
		Z_\mathcal{C}^{(i)}(T)&=Z_{\tau}^{(i)}(T)\frac{P_\mathcal{C}^{(i)}(T)}{p_{\tau,0}^{(i)}}\left(\frac{P_{\tau}^{(i)}(T)}{p_{\tau,0}^{(i)}}\right)^{-1}\\
		&=Z_{\tau}^{(i)}(T)\frac{P_\mathcal{C}^{(i)}(T)}{p_{\tau,0}^{(i)}}\sum_{u\geq 0} \left(\frac{1}{p_{\tau,0}^{(i)}}\right)^{u}\mathcal{P}_u\left(p_{\tau,0}^{(i)},-p_{\tau,1}^{(i)},-p_{\tau,2}^{(i)},\ldots,-p_{\tau,u}^{(i)}\right)T^u\\
		&=Z_{\tau}^{(i)}(T)\sum_{j\geq 0}p_j^{(i)}(\mathcal{C})T^j\sum_{u\geq 0} \left(\frac{1}{p_{\tau,0}^{(i)}}\right)^{u+1}\mathcal{P}_u\left(p_{\tau,0}^{(i)},-p_{\tau,1}^{(i)},-p_{\tau,2}^{(i)},\ldots,-p_{\tau,u}^{(i)}\right)T^u\\
		&=Z_{\tau}^{(i)}(T)\sum_{u\geq 0}T^u\sum_{j=0}^u\frac{p_j^{(i)}(\mathcal{C})}{\left(p_{\tau,0}^{(i)}\right)^{u-j+1}}\mathcal{P}_{u-j}\left(p_{\tau,0}^{(i)},-p_{\tau,1}^{(i)},-p_{\tau,2}^{(i)},\ldots,-p_{\tau,u-j}^{(i)}\right)\\
		&=Z_{\tau}^{(i)}(T)\sum_{u\geq 0}\beta_{\tau,u}^{(i)}(\mathcal{C})T^u,
	\end{align*}
as desired.
\end{proof}

An immediate consequence of the theorem above is the following.

\begin{corollary}
\label{CorollaryPbeta}
	Let $\beta_{\tau,u}^{(i)}(\mathcal{C})$, $u\geq 0$, be the coefficients defined in Theorem \ref{TheoremZBeta}. For all $0\leq \tau\leq m-1$ we have
	\begin{equation*}
		P_\mathcal{C}^{(i)}(T)= P_{\tau}^{(i)}(T)\sum_{u\geq 0}\beta_{\tau,u}^{(i)}(\mathcal{C})T^u.
	\end{equation*}
\end{corollary}

The next result gives an alternative way to compute the $\beta_{\tau,u}^{(i)}(\mathcal{C})$.

\begin{corollary}
	The following holds for all $u\geq 0$ and $0\leq \tau\leq m-1$.
	\begin{equation*}
		\beta_{\tau,u}^{(i)}(\mathcal{C}):=\sum_{j=0}^u\frac{b_j^{(i)}(\mathcal{C})}{\left(b_{\tau,0}^{(i)}\right)^{u-j+1}}\mathcal{P}_{u-j}\left(b_{\tau,0}^{(i)},-b_{\tau,1}^{(i)},-b_{\tau,2}^{(i)},\ldots,-b_{\tau,u-j}^{(i)}\right).
	\end{equation*}
\end{corollary}
\begin{proof}
	By Lemma~\ref{LemmaBell} we have
	\begin{equation*}
		\left(\frac{Z_\tau^{(i)}(T)}{b_{\tau,0}^{(i)}}\right)^{-1}=\sum_{u\geq 0} \left(\frac{1}{b_{\tau,0}^{(i)}}\right)^u\mathcal{P}_u\left(b_{\tau,0}^{(i)},-b_{\tau,1}^{(i)},-b_{\tau,2}^{(i)},\ldots,-b_{\tau,u}^{(i)}\right)T^u.
	\end{equation*}
	Theorem~\ref{TheoremZBeta} implies
	\begin{align*}
		\sum_{u\geq 0}\beta_{\tau,u}^{(i)}(\mathcal{C})T^u&=\frac{Z_\mathcal{C}^{(i)}(T)}{b_{\tau,0}^{(i)}}\left(\frac{Z_\tau^{(i)}(T)}{b_{\tau,0}^{(i)}}\right)^{-1}\\
		&=\sum_{j\geq 0}b_j^{(i)}(\mathcal{C})T^j\sum_{u\geq 0} \left(\frac{1}{b_{\tau,0}^{(i)}}\right)^{u+1}\mathcal{P}_u\left(b_{\tau,0}^{(i)},-b_{\tau,1}^{(i)},-b_{\tau,2}^{(i)},\ldots,-b_{\tau,u}^{(i)}\right)T^u\\
		&=\sum_{u\geq 0}T^u\sum_{j=0}^u\frac{b_j^{(i)}(\mathcal{C})}{\left(b_{\tau,0}^{(i)}\right)^{u-j+1}}\mathcal{P}_{u-j}\left(b_{\tau,0}^{(i)},-b_{\tau,1}^{(i)},-b_{\tau,2}^{(i)},\ldots,-b_{\tau,u-j}^{(i)}\right).
\qedhere	\end{align*}
\end{proof}

We conclude this section by showing that the $\beta_{\tau,u}^{(i)}(\mathcal{C})$ are indeed the coefficients of $W_\mathcal{C}^{(i)}(X,Y)$ with respect of the basis of the $i$-BMD rank weight enumerators.

\begin{corollary}
\label{CorollaryW=betaM}
	For all $0\leq \tau\leq m-1$ we have
	\begin{equation*}
		W_\mathcal{C}^{(i)}(X,Y)=\sum_{j=0}^{n-d_i}\beta_{\tau,j}^{(i)}(\mathcal{C})M_{\tau,d_i+j}^{(i)}(X,Y).
	\end{equation*}
\end{corollary}
\begin{proof}
The result follows combining Lemma~\ref{LemmaWivarphi} and Theorem~\ref{TheoremZBeta}, as
$W_\mathcal{C}^{(i)}(X,Y)$ is the coefficient of $T^{n-d_i}$ in the expression
	\begin{align*}
		Z_\mathcal{C}^{(i)}(T)\varphi(X,Y,T)&= Z_{\tau}^{(i)}(T)\varphi(X,Y,T)\sum_{u\geq 0}\beta_{\tau,u}^{(i)}(\mathcal{C})T^u\\
		&\equiv\sum_{j=0}^nM_{j,\tau}^{(i)}(X,Y)T^{n-j}\sum_{u\geq 0}\beta_{\tau,u}^{(i)}(\mathcal{C})T^u \mod T^n\\
		&\equiv\sum_{u=0}^nT^u\sum_{j=0}^u\beta_{\tau,j}^{(i)}(\mathcal{C})M_{\tau,n-u+j}^{(i)}(X,Y) \mod T^n.
\qedhere
	\end{align*}
\end{proof}

\begin{remark}
	If $i=1$ and $m|k$, then Corollary~\ref{CorollaryPbeta} and Remark~\ref{RemarkbuBMD} imply  that
	\begin{equation*}
		P_\mathcal{C}^{(1)}(T)=P_0^{(1)}(T)\sum_{u\geq 0}\beta_{0,u}^{(1)}(\mathcal{C})T^u=\qbin{m}{1}\sum_{u\geq 0}\beta_{0,u}^{(1)}(\mathcal{C})T^u.
	\end{equation*}
	In particular we have that
	\begin{equation*}
		\deg\left(\sum_{u\geq 0}\beta_{0,u}^{(i)}(\mathcal{C})T^u\right)=n-d-d^\perp+2 \qquad\textup{ and }\qquad p_u^{(1)}(\mathcal{C})=\qbin{m}{1}\beta_{0,u}^{(i)}(\mathcal{C})
	\end{equation*}
	for every $0\leq u\leq n-d-d^\perp+2$. Moreover, Corollary \ref{CorollaryW=betaM} gives
	\begin{equation*}
		W_\mathcal{C}^{(1)}(X,Y)=\qbin{m}{1}^{-1}\sum_{j=0}^{n-d}p_u^{(1)}(\mathcal{C})M_{0,d+j}^{(1)}(X,Y),
	\end{equation*}
	which is the second part of \cite[Theorem~1]{blanco2018rank} up to a constant.
\end{remark}

We conclude this section illustrating
Theorem~\ref{TheoremZBeta} and Corollary~\ref{CorollaryW=betaM}
in an example with $i=3$ and $\tau=1$.

\begin{example}
	Let $\mathcal{C}_1$ be the code of Example~\ref{Example(i)not(i+1)MRD}. Recall that $d_3(\mathcal{C}_1)=2$. We have
	\begin{equation*}
		\begin{aligned}
			Z_{\mathcal{C}_1}^{(3)}(T)&=\frac{1}{7}+1395T+6347715T^2+26167664835T^3+107225699266755T^4+\ldots,\\
			W_{\mathcal{C}_1}^{(3)}(T)&=XY^2 + 1394Y^3,
	\end{aligned}
	\end{equation*}
	and\footnote{ Let $\mathcal{C}_6$ be the minimally $3$-BMD code of Section~\ref{SectioniMRD} of dimension $5$. For $3\leq j\leq 5$ we have $Z_1^{(j)}(T)=Z_{\mathcal{C}_6}^{(j)}(T)$, since $\dim(\mathcal{C}_6)\equiv 1 \mod 4$.}
	\begin{equation*}
	\begin{aligned}
			Z_1^{(3)}(T)&= 155+788035T+3269560515T^2+13402854502595T^3+\ldots,\\
			M_{1,2}^{(3)}(X,Y)&=1085XY^2 + 786950Y^3,\\
			M_{1,3}^{(3)}(X,Y)&=155Y^3.
		\end{aligned}
	\end{equation*}	
	On the other hand, since
	\begin{equation*}
	\begin{aligned}
		\mathcal{P}_0(x_0)=1, \qquad \mathcal{P}_1(x_0,x_1)=x_1, \qquad \mathcal{P}_2(x_0,x_1,x_2)=x_1^2+x_2x_0,\\ \mathcal{P}_3(x_0,x_1,x_2,x_3)=x_1^3+2x_1x_2x_0+x_3x_0^2, \qquad \ldots\qquad\qquad
	\end{aligned}
	\end{equation*}
	we get
	\begin{equation*}
		\begin{tabular}{ccccc}
			$\displaystyle\beta_{1,0}^{(3)}(\mathcal{C}_1)=\frac{1}{1085}$, & $\quad$ &$\displaystyle\beta_{1,1}^{(3)}(\mathcal{C}_1)=\frac{145108}{33635}$,& \\\\$\displaystyle\beta_{1,2}^{(3)}(\mathcal{C}_1)=-\frac{440232944}{1042685}$, & $\quad$ &$\displaystyle\beta_{1,3}^{(3)}(\mathcal{C}_1)=\frac{928753518821747}{6464647}$,& $\quad$ & \ldots
		\end{tabular}
	\end{equation*}
	One can check that
	\begin{equation*}
		Z_{\mathcal{C}_1}^{(3)}(T)=Z_1^{(2)}(T)\sum_{u\geq 0}\beta_{1,u}^{(3)}(\mathcal{C}_1)T^u,
	\end{equation*}
	\begin{equation*}
		W_{\mathcal{C}_1}^{(3)}(T)=\beta_{1,0}^{(3)}(\mathcal{C}_1)M_{1,2}^{(3)}(X,Y)+\beta_{1,1}^{(3)}(\mathcal{C}_1)M_{1,3}^{(3)}(X,Y).
	\end{equation*}
\end{example}

\section{Duality Results}
\label{SectionDuality}
In this section we establish a MacWilliams identity for the generalized rank distribution. We then use it to explicitly compute the generalized zeta functions, the generalized binomial moments and the generalized rank weight distributions of $\mathcal{C}^\perp$ knowing the generalized binomial moments of $\mathcal{C}$. The following result generalizes \cite[Lemma~30]{ravagnani2016rank}.

\begin{theorem}
\label{TheoremMacBu}
	The following holds for all $0\leq u\leq n$.
	\begin{equation*}
		B_u^{(i)}(\mathcal{C})=\sum_{j=0}^iq^{j(k-m(n-u)-i+j)}\qbin{k-m(n-u)}{i-j}B_{n-u}^{(j)}(\mathcal{C}^\perp).
	\end{equation*}
\end{theorem}
\begin{proof}
Suppose first that $n=m$ and $i=0$. Then Remark \ref{RemarkBuiAnt} implies that 
	\begin{equation*}
		B_{u}^{(0)}(\mathcal{C})=\qbin{n}{u}=B_{n-u}^{(0)}(\mathcal{C}^\perp)=\sum_{j=0}^0q^{j(k-m(n-u)-i+j)}\qbin{k-m(n-u)}{i-j}B_{n-u}^{(j)}(\mathcal{C}^\perp).
	\end{equation*}	
Now assume $i \ge 1$ and $n \le m$. Define 
	\begin{equation*}
		\gamma:=\begin{cases}
			1 &\;\textup{ if }\; n<m \;\textup{ or }\;u=n=m, \\[8pt]
			\displaystyle\frac{1}{2} &\;\textup{ if }\; n=m\;\textup{ and }\;  1\leq u\leq n-1.
		\end{cases}
	\end{equation*}
By Lemma \ref{LemmaCcapA} we have
	\begin{align*}
		B_u^{(i)}(\mathcal{C})&=\gamma\sum_
		{\tiny\begin{matrix}
			A\in \mathcal{A}_q(n,m)\\
			\dimq{A}=mu
		\end{matrix}}
		\qbin{\dimq{\mathcal{C}\cap A}}{i}\\
		&=\gamma\sum_
		{\tiny\begin{matrix}
			A\in \mathcal{A}_q(n,m)\\
			\dimq{A}=mu
		\end{matrix}}
		\qbin{\dimq{\mathcal{C}^\perp\cap A^\perp}+k-m(n-u)}{i}\\
		&=\gamma\sum_
		{\tiny\begin{matrix}
			A\in \mathcal{A}_q(n,m)\\
			\dimq{A}=mu
		\end{matrix}}\sum_{j=0}^iq^{j(k-m(n-u)-i+j)}\qbin{k-m(n-u)}{i-j}\qbin{\dimq{\mathcal{C}^\perp\cap A^\perp}}{j}\\
		&=\sum_{j=0}^iq^{j(k-m(n-u)-i+j)}\qbin{k-m(n-u)}{i-j}\gamma\sum_
		{\tiny\begin{matrix}
			A\in \mathcal{A}_q(n,m)\\
			\dimq{A}=mu
		\end{matrix}}\qbin{\dimq{\mathcal{C}^\perp\cap A^\perp}}{j},
	\end{align*}
where in the third line we applied Lemma \ref{LemmaTools}.
	Since the map $A\mapsto A^\perp$ is a bijection between the $mu$-dimensional and the $m(n-u)$-dimensional optimal anticodes in $\Fqnm$, we get
	\begin{equation*}
		\gamma\sum_
		{\tiny\begin{matrix}
			A\in \mathcal{A}_q(n,m)\\
			\dimq{A}=mu
		\end{matrix}}\qbin{\dimq{\mathcal{C}^\perp\cap A^\perp}}{j}=\ \gamma\sum_
		{\tiny\begin{matrix}
			A\in \mathcal{A}_q(n,m)\\
			\dimq{A}=m(n-u)
		\end{matrix}}\qbin{\dimq{\mathcal{C}^\perp\cap A}}{j}=B_{n-u}^{(j)}(\mathcal{C}^\perp),
	\end{equation*}
	from which the statement follows.
\end{proof}

\begin{example}
	Let $\mathcal{C}_1$ be the code of Example \ref{Example(i)not(i+1)MRD} and $\mathcal{C}_1^\perp$ its dual.  Fix $i=u=2$, then we have $B_2^{(2)}(\mathcal{C})=13$ and
	\begin{equation*}
		B_1^{(0)}(\mathcal{C}^\perp)=7, \qquad B_1^{(1)}(\mathcal{C}^\perp)=1, \qquad B_1^{(2)}(\mathcal{C}^\perp)=0.
	\end{equation*}
	One can check that
	\begin{equation*}
			B_2^{(2)}(\mathcal{C})=\qbintwo{2}{2}B_1^{(0)}(\mathcal{C}^\perp)+2\qbintwo{2}{1}B_1^{(1)}(\mathcal{C}^\perp)+2^4\qbintwo{2}{0}B_1^{(2)}(\mathcal{C}^\perp).
	\end{equation*}
\end{example}

\begin{corollary}
	The following hold for all $0\leq u\leq n-d_i-d$.
		\begin{enumerate}
			\item \label{s1} $\displaystyle b_u^{(i)}(\mathcal{C})=\sum_{j=0}^iq^{j(k-m(n-d_i-u)-i+j)}\qbin{k-m(n-d_i-u)}{i-j}b_{n-u-d_j^\perp-d_i}^{(j)}(\mathcal{C}^\perp)$,
			\item \label{s2} $\displaystyle Z_\mathcal{C}^{(i)}(T)=T^{n-d_i}\sum_{j=0}^i\frac{q^{j(k-md_j^\perp-i+j)}}{T^{d_j^\perp}}\sum_{t\leq n-d_j^\perp-d_i}\frac{\qbin{k-m(t+d_j^\perp)}{i-j}}{q^{jmt}T^t}b_t^{(j)}(\mathcal{C}^\perp)$.
		\end{enumerate}
\end{corollary}
\begin{proof}
	Part \ref{s1} follows from Theorem \ref{TheoremMacBu} and the definition of $b_u^{(i)}(\mathcal{C})$.
Part \ref{s2} is a consequence of applying the first equation to the definition of generalized zeta function.
\end{proof}

\begin{corollary}
	For all $0\leq w\leq n$ we have
	\begin{equation*}
		A_w^{(i)}(\mathcal{C})=\sum_{u=0}^w(-1)^{w-u}q^{\binom{w-u}{2}}\qbin{n-u}{n-w}\sum_{j=0}^i\frac{q^{j(k+mu+j)}}{q^{mn+i}}\qbin{k-m(n-u)}{i-j}\sum_{t=0}^{n-u}\qbin{n-t}{u}A_t^{(j)}(\mathcal{C}^\perp).
	\end{equation*}
\end{corollary}
\begin{proof}
	The desired formula follows from Theorem \ref{TheoremBuAw} and Theorem \ref{TheoremMacBu}.
\end{proof}

\section{Codes for the Hamming Metric}
\label{SectionHamming}
The results of the previous sections have analogues for the Hamming metric, with the main distinction given by the notion of support of a code and its codewords. The two theories are linked by the role played by the two associated lattices, in the sense of~\cite{ravagnani2018duality}. More precisely, in the rank metric this lattice is that of subspaces of a linear space over 
$\mathbb{F}_q$, while in the Hamming metric it is the Boolean algebra over the set $\{1,...,n\}$. 

In the sequel we use the characterization of generalized Hamming weights given in~\cite[Proposition 9]{ravagnani2016generalized} for $q \ge 3$. All the results in this section are stated under this assumption. 

We recall that a linear code is a subspace of $\Fqn$ of dimension $k$ whose Hamming weight and support are defined respectively as $\wt(c):=|\{1 \le i \le n : c_i\ne 0\}|$ and $\supp(C):=\{1 \le i \le n : \exists\; c\in C \textup{ with } c_i\ne 0 \}$. 
Using the characterization of~\cite{ravagnani2016generalized} we define the \textbf{$i$-th generalized weight} of $C$ as 
\begin{equation*}
	d_i(C):=\min\{\dimq{A}:A\in\mathcal{A}(n),\dimq{C\cap A}\geq i\},
\end{equation*}
	where $A(n)$ is the set of \textbf{optimal linear anticodes} in the Hamming metric, i.e., the set of all the codes such that~$\dimq{C}=\max\{\wt(c): c \in C\}$. 

The property of being $i$-MDS has already appeared in literature~\cite[Section~6]{wei1991generalized}. For a $\Fq$-$[n,k,d]$ Hamming-metric code $C$ we say that $C$ is $i$-\textbf{BMD} if $n-d^\perp-d_i(C)<0$. Note that if $C$ is $i$-BMD then it is also $(i+1)$-BMD. Indeed,
\begin{equation*}
	n-d^\perp-d_{i+1}(C)<n-d^\perp-d_{i}(C)<0.
\end{equation*}

As in Section \ref{SectionBin}, a useful property of $i$-BMD codes is that for any $i\leq j\leq k$, their $j$-th generalized rank weight distributions and binomial moments depend only on some fundamental code parameters. More precisely, if $C$ is a minimally $i$-BMD code then the following hold for all $i\leq j\leq k$:
	\begin{enumerate}
		\item $\displaystyle b_u^{(j)}(C):=\begin{cases}
					0 &\;\textup{ if }\; u<0,\\
					\qbin{k-n+u+d_j(C)}{j}  &\;\textup{ if }\; u\geq 0,
				\end{cases}$
		\item $\displaystyle A_w^{(j)}(C)=\bin{n}{w}\sum_{u=d_j(C)}^w(-1)^{w-u}\bin{w}{u}\qbin{k-n+u}{j}$,
	\end{enumerate}
where $\left(A_0^{(j)}(C), A_1^{(j)}(C),\ldots,A_n^{(j)}(C)\right)$ denotes the $j$-th weight distribution of $C$.

We conclude this section with the Hamming metric analogue of Theorem~\ref{TheoremiBMD=>iMRD}.
One implication can be obtained as in the rank metric case, using Wei's duality for the
Hamming-metric generalized weights~\cite[Theorem~3]{wei1991generalized}. 
The other implication also follows from Wei's duality and is omitted here.

\begin{theorem}
	$C$ is $i$-BMD if and only if $C$ is $i$-MDS.
\end{theorem}

\bigskip

\bibliographystyle{siam}  
\bibliography{references}

\end{document}